\documentclass[12pt]{article}
\usepackage[margin=1in]{geometry}
\usepackage{amsmath,amssymb,amsthm}
\newcommand{\E}{\mathbb{E}}
\usepackage{graphicx}
\usepackage{booktabs}
\usepackage{natbib}
\usepackage{hyperref}
\usepackage{algorithm}
\usepackage{algorithmic}
\usepackage{multirow}
\usepackage{tikz}
\usetikzlibrary{arrows.meta, positioning}

\newtheorem{proposition}{Proposition}
\newtheorem{remark}{Remark}

\title{The Negative Binomial Chain--Ladder:\\A Full Likelihood Model for Claim Count Reserving}

\author{Robin Van Oirbeek\thanks{University of Antwerp, Department of Mathematics, Antwerp, Belgium. Email: \texttt{robin.vanoirbeek@gmail.com}.}}

\date{\today}

\begin{document}
\maketitle

\begin{abstract}
The Chain--Ladder (CL) method remains the dominant macro-level technique for claims reserving in non-life insurance, yet its classical formulation lacks a coherent probabilistic foundation. Existing stochastic extensions such as the Mack model and the Over-Dispersed Poisson (ODP) framework provide measures of uncertainty but rely on second-moment assumptions or quasi-likelihood variance structures without clear generative interpretations.

This paper develops a Negative Binomial Chain--Ladder (NB--CL) model that embeds the CL method within a full likelihood-based framework. The key contribution is a micro-level derivation showing that the negative binomial distribution arises naturally from a Poisson--Gamma construction: claims arrive according to a Poisson process whose cell-level rate carries a multiplicative Gamma shock, and marginalisation yields negative binomial incremental counts. This derivation gives the dispersion parameter $\kappa$ a structural interpretation as the variability of cell-level claim-generating conditions around the systematic accident-year and development structure, rather than an ad-hoc overdispersion adjustment; we show, conversely, that heterogeneity placed one level up, i.e.\ a frailty shared by an entire accident year, is absorbed by the free accident-year parameters and is not identifiable from a single triangle.

The NB--CL model generalises the Poisson Chain--Ladder model, recovering it exactly in the limit $\kappa \to \infty$; its point estimates coincide with the classical Chain--Ladder (and ODP) estimates in that limit and approximate them closely for finite $\kappa$, while its variance function differs structurally from the ODP one (quadratic vs.\ linear). The model unifies the cross-classified Chain--Ladder family within a single probabilistic hierarchy. Note that alternative likelihood foundations for ODP exist, including the NB1 parameterisation and W\"uthrich's scaled Poisson representation; the NB2 formulation adopted here implies a quadratic variance-mean relationship. A parametric bootstrap procedure is developed to incorporate both process and parameter uncertainty. Simulation studies confirm near-nominal coverage under correct specification once the dispersion parameter is bias-corrected, and a controlled degradation under model misspecification. Empirical illustrations on claim count data (Australian motor bodily injury) and paid amounts (Taylor--Ashe) document both the structural reading of $\kappa$ and the working-approximation status of the model in the amounts case.

\medskip
\noindent\textbf{Keywords:} claims reserving, Chain--Ladder, negative binomial, overdispersion, GLM, stochastic reserving, IBNR
\end{abstract}

\newpage

\section{Introduction}

Claims reserving is a central task in non-life insurance \citep[see][for textbook treatments]{taylor2000,england2002}. At each valuation date, insurers must estimate the liabilities associated with claims that have occurred but are not yet fully settled. These liabilities include both Reported But Not Settled (RBNS) claims and Incurred But Not Reported (IBNR) claims. While micro-level models can analyse the development of individual claims, macro-level reserving techniques remain dominant in practice due to their simplicity, transparency, and robustness.

Among macro-level methods, the Chain--Ladder (CL) technique is by far the most widely used. Its appeal lies in its deterministic structure: development factors are estimated from historical run-off triangles and applied multiplicatively to project future development. Despite its practical importance, the classical CL method lacks a coherent probabilistic foundation. It provides point estimates but no likelihood, and its stochastic extensions rely on assumptions that do not always align with the underlying claim process.

Two prominent stochastic extensions are the Mack model \citep{mack1993,mack1994} and the Over-Dispersed Poisson (ODP) framework \citep{renshaw1998,england2002}. The Mack model specifies only conditional first and second moments of the cumulative claims and does not assume a full probability distribution; prediction uncertainty is obtained from analytic mean-squared-error formulae rather than from a predictive distribution. The ODP model provides a quasi-likelihood interpretation but relies on a variance structure of the form $\text{Var}(N_{i,j}) = \phi\, \text{E}[N_{i,j}]$, which may not reflect the true variability in the data; its predictive distribution is typically obtained via residual resampling \citep{england2002}, which introduces additional assumptions and may not be fully consistent with the underlying claim process.

Recent work has independently pursued full-likelihood negative binomial
models for count triangles. \citet{nietobarajas2026} propose Bayesian
negative binomial models whose primary contribution is to relax the
independence assumption across development years via moving-average
dependence sequences with negative binomial marginals. Our aims are
complementary but distinct: we work in a frequentist GLM framework and,
rather than modelling cross-period dependence, derive the negative
binomial from a micro-level Poisson--Gamma arrival process. This derivation allows us to pin down the identification status of $\kappa$ as a cell-level dispersion parameter (Remark~\ref{rem:absorption})
and recover the classical Chain--Ladder point estimates exactly in the
Poisson limit.

\subsection{Contribution}

This paper develops a Negative Binomial Chain--Ladder (NB--CL) model that provides a full likelihood for the CL method. While likelihood-based formulations of the Chain--Ladder method exist in the literature \citep{renshaw1998,verrall2000}, they typically treat overdispersion as a statistical nuisance parameter without structural interpretation.

The key contribution of this paper is a micro-level derivation showing that the negative binomial distribution arises naturally from a Poisson--Gamma construction. Claims in cell $(i,j)$ arrive according to a Poisson process whose rate carries a multiplicative Gamma shock, $N_{i,j} \mid \varepsilon_{i,j} \sim \mathrm{Poisson}(\exp(\alpha_i)\, w_j\, \varepsilon_{i,j})$ with $\varepsilon_{i,j} \sim \mathrm{Gamma}(\kappa, \kappa)$ independent across cells. Marginalising over the latent shock yields a negative binomial distribution for incremental counts.

This derivation gives the dispersion parameter $\kappa$ a structural interpretation as the variability of cell-level claim-generating conditions, i.e.\ period-specific environment shocks around the systematic accident-year and development structure, rather than as an ad-hoc overdispersion adjustment, and the constancy of $\kappa$ across cells is the assumption of a homogeneous shock distribution.

Heterogeneity one level up, i.e.\ a single frailty shared by an entire accident year, is natural to consider but, as Remark~\ref{rem:absorption} shows, it is structurally invisible in a model with free accident-year parameters: the $\alpha_i$ absorb any realised accident-year frailty, and its dispersion is not identifiable from a single triangle. The estimable $\kappa$ is the cell-level dispersion; identification of year-level heterogeneity requires anchored levels.

Beyond this micro-level foundation, the NB--CL model provides a unifying probabilistic framework for the entire Chain--Ladder family. The Poisson CL model is an exact special case of the NB--CL likelihood, the ODP model a locally equivalent quasi-likelihood approximation, and the Mack model a non-nested specification sharing its point estimates; the residual bootstrap \citep{england2002} is a simulation-based approximation of its predictive distribution. By deriving the CL structure from a micro-level Poisson--Gamma process, the NB--CL model clarifies the assumptions underlying classical reserving techniques and places them within a single coherent statistical framework.

\subsection{Outline}

The remainder of the paper is structured as follows. Section~2 introduces the claims reserving problem and notation. Section~3 presents the NB--CL model. Section~4 derives the model from a micro-level Poisson--Gamma construction. Section~5 discusses estimation, including model selection and practical considerations for sparse triangles. Section~6 develops predictive distributions incorporating both process and parameter uncertainty. Section~7 relates the NB--CL model to existing stochastic reserving approaches and presents a model hierarchy. Section~8 presents simulation results under both correct specification and model misspecification. Section~9 provides empirical illustrations on both claim count and paid amount data. Section~10 discusses implications and Section~11 concludes.

\section{The Claims Reserving Problem}

\subsection{Run-off triangles}

Non-life insurance claims typically evolve over multiple stages between the occurrence of the underlying event and the final settlement. At each valuation date, the insurer must estimate the outstanding liabilities associated with all claims that have occurred prior to that date.

The standard data structure for macro-level reserving is the run-off triangle. Let $N_{i,j}$ denote the incremental claim count (or amount) for accident year $i$ and development year $j$, where $i = 1, \ldots, I$ and $j = 0, \ldots, J-1$. For simplicity, we assume $I = J$, so the triangle is square. At the valuation date, only the upper-left portion of the triangle is observed:
\[
\mathcal{D} = \{N_{i,j} : i + j \leq I\}.
\]
The goal is to predict the lower-right portion:
\[
\mathcal{F} = \{N_{i,j} : i + j > I\},
\]
and to estimate the total reserve $R^{(N)}$ for outstanding counts as:
\[
R^{(N)} = \sum_{i+j > I} N_{i,j}.
\]

\subsection{The deterministic Chain--Ladder method}

The classical Chain--Ladder method estimates cumulative development factors from the observed triangle and applies them multiplicatively to project future development. Let $C_{i,j} = \sum_{l=0}^{j} N_{i,l}$ denote cumulative claims. The development factor for development year $j$ is
\[
\hat{f}_j = \frac{\sum_{i=1}^{I-j} C_{i,j+1}}{\sum_{i=1}^{I-j} C_{i,j}}, \quad j = 0, \ldots, J-2.
\]
The projected ultimate for accident year $i$ is
\[
\hat{C}_{i,J-1} = C_{i,I-i} \cdot \prod_{l=I-i}^{J-2} \hat{f}_l,
\]
and the total reserve estimate is $\hat{R}^{(N)} = \sum_{i=2}^{I} (\hat{C}_{i,J-1} - C_{i,I-i})$.

The deterministic CL method provides point estimates but no measure of uncertainty. The NB--CL model developed in this paper provides a probabilistic foundation that preserves these point estimates exactly in the Poisson limit and approximates them closely for finite $\kappa$, while enabling coherent uncertainty quantification.

\section{The Negative Binomial Chain--Ladder Model}\label{sec:model}

\subsection{Model specification}

The Negative Binomial Chain--Ladder (NB--CL) model assumes that incremental claim counts follow a negative binomial distribution with a log-additive mean structure:
\begin{equation}\label{eq:nbcl}
N_{i,j} \sim \text{NegBin}(\mu_{i,j}, \kappa), \quad \log \mu_{i,j} = \alpha_i + \beta_j,
\end{equation}
where $\mu_{i,j} = \exp(\alpha_i + \beta_j)$ is the expected count for cell $(i,j)$, $\alpha_i$ is the accident-year effect, $\beta_j$ is the development-year effect, and $\kappa > 0$ is the dispersion parameter controlling overdispersion. We adopt the simplex parameterisation
\[
  \sum_{j=0}^{J-1} \exp(\beta_j) = 1,
\]
which gives the development-year effects a direct probabilistic interpretation: $w_j := \exp(\beta_j)$ is the proportion of ultimate claims reported in development year $j$, with $\sum_j w_j = 1$. Under this convention, $\mu_i := \exp(\alpha_i)$ is the expected total claim count for accident year $i$, and the cell mean factorises as $\mu_{i,j} = \mu_i w_j$.

We use the mean-dispersion parameterisation of the negative binomial distribution, under which
\[
\text{E}[N_{i,j}] = \mu_{i,j}, \quad \text{Var}(N_{i,j}) = \mu_{i,j} + \frac{\mu_{i,j}^2}{\kappa}.
\]
The variance exceeds the Poisson variance $\mu_{i,j}$ by a factor that depends on $\kappa$. As $\kappa \to \infty$, the variance approaches the Poisson case; as $\kappa \to 0$, overdispersion becomes extreme.

\subsection{Log-additive structure and Chain--Ladder}

The log-additive structure $\log \mu_{i,j} = \alpha_i + \beta_j$ implies
\[
\mu_{i,j} = \exp(\alpha_i) \cdot \exp(\beta_j),
\]
a multiplicative decomposition into accident-year and development-year components. This two-way cross-classified structure connects the Chain--Ladder method to the analysis of variance framework \citep{kremer1982} and is fundamental to GLM-based reserving.

The ratio of expected cumulative claims between successive development years is
\[
\frac{\sum_{l=0}^{j} \mu_{i,l}}{\sum_{l=0}^{j-1} \mu_{i,l}} = \frac{\sum_{l=0}^{j} \exp(\beta_l)}{\sum_{l=0}^{j-1} \exp(\beta_l)},
\]
which is constant across accident years. This corresponds exactly to the Chain--Ladder development factor for development year $j$.

\subsection{The dispersion parameter $\kappa$}

The dispersion parameter $\kappa$ is best interpreted via the micro-level derivation in Section~\ref{sec:micro}: it represents the variability of cell-level claim-generating conditions around the systematic structure.

Large values of $\kappa$ (e.g., $\kappa > 50$) indicate weak cell-level shocks, with variance close to Poisson. Moderate values (e.g., $5 < \kappa < 50$) represent the typical range for many portfolios with meaningful overdispersion. Small values (e.g., $\kappa < 5$) indicate strong cell-level shocks, i.e.\ pronounced departures of individual cells from the multiplicative structure, possibly signalling pattern instability, unmodelled calendar effects, or model misspecification.

\section{Micro-Level Derivation}\label{sec:micro}

The NB--CL model is the macro-level aggregation of a Poisson--Gamma process with multiplicative development weights. 

\subsection{Micro-level claim arrival process}

Consider a portfolio in which individual claims occur according to a Poisson process. Let $\lambda_i$ denote the underlying claim arrival rate for accident year $i$. Conditional on $\lambda_i$, the number of claims reported in development year $j$ follows a Poisson distribution:
\begin{equation}\label{eq:poisson}
N_{i,j} \mid \lambda_i \sim \text{Poisson}(\lambda_i w_j),
\end{equation}
where $w_j > 0$ is a development weight associated with development year $j$. These weights capture the reporting pattern and satisfy $\sum_{j=0}^{J-1} w_j = 1$.

The multiplicative structure of the Chain--Ladder method is recovered by setting $w_j = \exp(\beta_j)$, which under the simplex constraint $\sum_j \exp(\beta_j) = 1$ of Section~\ref{sec:model} satisfies $\sum_j w_j = 1$ as required.

\begin{remark}[Ingredients for the Chain--Ladder estimators]
Three assumptions together are sufficient to recover the classical
Chain--Ladder development factors from the NB--CL model.
\begin{enumerate}
  \item \emph{Multiplicative mean structure.} The log-additive structure
        $\log \mu_{i,j} = \alpha_i + \beta_j$ ensures that the ratio of
        expected cumulative claims between successive development years is
        constant across accident years, which is the defining property of
        the Chain--Ladder development factor.
  \item \emph{Conditional independence across development years.}
        The cell shocks $\varepsilon_{i,j}$ are independent, so the incremental counts $N_{i,0}, \ldots, N_{i,J-1}$ are independent, both conditionally and marginally. Without this independence, the column-wise accumulation underlying the development factors would not be well-defined. 
  \item \emph{Poisson--Gamma hierarchy.} The Gamma prior on $\lambda_i$
        yields the Negative Binomial marginal and the full likelihood. The
        multiplicative mean structure alone, or a different prior on
        $\lambda_i$, would give a different marginal distribution and a
        different likelihood, even if point estimates coincidentally agreed
        for a particular parameterisation.
\end{enumerate}
The first two ingredients suffice to recover the Chain--Ladder \emph{point
estimates} when estimation proceeds by Poisson or quasi-Poisson scoring,
which solves the marginal-totals equations; under the negative binomial
likelihood with finite $\kappa$ the working weights differ from the Poisson
weights and the estimates depart slightly from the classical ones
(Section~8.2). The third ingredient is required for the full likelihood and
the coherent predictive distribution developed in Section~\ref{sec:pred}.
\end{remark}

\subsection{Gamma shocks at the cell level}

Beyond the systematic accident-year and development effects, individual cells depart from the fitted multiplicative structure because the claim-generating environment fluctuates from one (accident-year, development-year) period to the next: short-lived operational, reporting, or exposure conditions local to a cell rather than shared across a row. To capture this residual cell-level variability, we equip each cell with a multiplicative Gamma shock:
\begin{equation}\label{eq:gamma}
N_{i,j} \mid \varepsilon_{i,j} \sim
\mathrm{Poisson}\!\big(\exp(\alpha_i)\, w_j\, \varepsilon_{i,j}\big),
\qquad
\varepsilon_{i,j} \sim \mathrm{Gamma}(\kappa, \kappa) \ \text{i.i.d.},
\end{equation}
with $\mathrm{E}[\varepsilon_{i,j}] = 1$ and
$\mathrm{Var}(\varepsilon_{i,j}) = 1/\kappa$. Marginalising over the
shock yields $N_{i,j} \sim \mathrm{NegBin}(\mu_{i,j}, \kappa)$ with
$\mu_{i,j} = \exp(\alpha_i) w_j$, independently across cells, which is
exactly the likelihood~\eqref{eq:loglik} maximised in Section~5.

The Gamma distribution is a natural choice for modelling heterogeneity 
in Poisson rates, both because it is conjugate to the Poisson and 
because it leads to a closed-form marginal distribution. This 
Poisson--Gamma structure has deep connections to credibility theory 
\citep{buhlmann1970}, where Gamma-distributed random effects capture 
heterogeneity across risk classes. \citet{gisler2008} develop this 
connection specifically for claims reserving, showing that Bühlmann--Straub 
credibility estimators provide a natural interpretation of Chain--Ladder 
development factors and yield refined uncertainty estimates by combining 
individual accident-year information with portfolio-level pooling. 
A credibility reading of the Chain--Ladder family requires a hierarchical level structure, with accident-year levels drawn around a common centre or anchored by exposure, so that information pools
across rows. The NB--CL model as specified here, with a free
$\alpha_i$ per accident year, deliberately performs no such pooling:
each row's level is estimated from that row alone, and the Gamma
layer of~\eqref{eq:gamma} operates at the cell level. The likelihood-based credibility counterpart of \citet{gisler2008} hence only requires a common-rate, exposure-anchored hierarchy on the levels, and is developed in \citet{vanoirbeek2026cred}.

\begin{proposition}[Conjugacy of the Gamma shock]\label{prop:gamma}
Among unit-mean mixing distributions on $(0,\infty)$, the Gamma family
is the conjugate family for Poisson sampling \citep{diaconis1979}.
Conjugacy closes the mixing integral in negative binomial form, places the marginal in the exponential dispersion family, hereby enabling profile-likelihood estimation of $\kappa$ within a standard GLM, and yields the linear posterior mean underlying credibility formulae
\citep{jewell1974}.
\end{proposition}

Non-conjugate mixing distributions lose the GLM machinery, not
identifiability. The lognormal mixture has no closed-form marginal; the inverse Gaussian has a closed-form marginal \citep{willmot1987} that lies outside the exponential dispersion family, so $\kappa$ is then estimated by direct numerical likelihood rather than via \texttt{glm.nb}. Note that the inverse Gaussian belongs to the natural exponential family with power variance function (Tweedie power $p = 3$): the tractability of the Gamma choice is a consequence of Poisson \emph{conjugacy} specifically, not of NEF-PVF membership.

\subsection{Marginal distribution of incremental counts}

Marginalising over the cell-level shock $\varepsilon_{i,j}$ yields the
distribution of $N_{i,j}$:
\[
N_{i,j} \;=\; \int_0^\infty
\mathrm{Poisson}\!\bigl(\exp(\alpha_i)\, w_j\, \varepsilon\bigr)\cdot
\mathrm{Gamma}(\varepsilon;\, \kappa, \kappa)\, d\varepsilon.
\]

This integral is the well-known Poisson--Gamma mixture, which yields a negative binomial distribution:
\begin{equation}\label{eq:marginal}
N_{i,j} \sim \text{NegBin}(\mu_{i,j}, \kappa), \quad \mu_{i,j} = \exp(\alpha_i) \cdot w_j.
\end{equation}
Taking logarithms and absorbing the normalisation into $\beta_j$:
\[
\log \mu_{i,j} = \alpha_i + \beta_j,
\]
which recovers the NB--CL model specification~\eqref{eq:nbcl}.

\begin{proposition}[Poisson--Gamma mixture]
Let $N \mid \lambda \sim \text{Poisson}(\lambda w)$ and $\lambda \sim \text{Gamma}(\kappa, \kappa/\mu)$. Then $N \sim \text{NegBin}(\mu w, \kappa)$ with $\text{E}[N] = \mu w$ and $\text{Var}(N) = \mu w + (\mu w)^2/\kappa$.
\end{proposition}

\begin{remark}[Cumulative counts are also conditionally Poisson]
Conditional on the cell-shock vector
$(\varepsilon_{i,0},\ldots,\varepsilon_{i,j})$, the increments $N_{i,l}$ are
independent Poisson variables, so by additivity the cumulative count
$C_{i,j}=\sum_{l=0}^{j}N_{i,l}$ satisfies
\[
C_{i,j}\mid(\varepsilon_{i,l})_{l\le j}\;\sim\;
\mathrm{Poisson}\!\Bigl(\exp(\alpha_i)\sum_{l=0}^{j}w_l\,\varepsilon_{i,l}\Bigr).
\]
The cumulative triangle is thus conditionally Poisson, with the cumulative
weighted shock $\sum_{l\le j}w_l\varepsilon_{i,l}$ replacing the single cell
weight. Because the shocks have unit mean, the marginal expectation is
\[
\E[C_{i,j}]=\exp(\alpha_i)\sum_{l=0}^{j}w_l=\exp(\alpha_i)\,W_j,
\qquad W_j=\sum_{l\le j}w_l,
\]
and the expected age-to-age factor
\[
\frac{\E[C_{i,j+1}]}{\E[C_{i,j}]}=\frac{W_{j+1}}{W_j}
\]
is constant across accident years, the defining property of the
Chain--Ladder development factor. The constancy is a property of the
\emph{marginal} mean: for a fixed shock realisation the ratio
$\bigl(\sum_{l\le j+1}w_l\varepsilon_{i,l}\bigr)/
\bigl(\sum_{l\le j}w_l\varepsilon_{i,l}\bigr)$ depends on $i$ through the
realised shocks, and the common development factor emerges only once the
unit-mean shocks are averaged out. The CL formula is therefore not an
algebraic coincidence but follows from the conditional-Poisson cell structure
together with the unit-mean normalisation of the shocks.
\end{remark}

\subsection{Interpretation of $\kappa$}

The dispersion parameter $\kappa$ controls the variability of
cell-level conditions around the systematic multiplicative structure.
As $\kappa \to \infty$ the shocks degenerate and the model reduces to
the Poisson Chain--Ladder; for finite $\kappa$,
$\mathrm{Var}(N_{i,j}) = \mu_{i,j} + \mu_{i,j}^2/\kappa$, which
exceeds the Poisson variance and reflects overdispersion due to
cell-level shocks.

\begin{remark}[Accident-year frailty is absorbed by the $\alpha_i$]
\label{rem:absorption}
Replacing the cell-level shocks in~\eqref{eq:gamma} by a single
frailty shared across an accident year, i.e.\ $\lambda_i \sim \mathrm{Gamma}(\kappa, \kappa/\exp(\alpha_i))$ with cells conditionally independent given $\lambda_i$, produces the same
negative binomial \emph{marginals} but a different joint law, whose
likelihood factorises into a negative binomial row total and a
frailty-free multinomial split. With free $\alpha_i$ the realised
frailty is absorbed into the fitted accident-year parameter, and the
frailty dispersion is not identifiable from a single triangle. The reason is
visible in the factorisation: the multinomial split carries no frailty, since a
row-level frailty scales all cells of a row equally and cancels in the
conditional split, while the negative binomial row-total component is fitted
exactly by the free $\alpha_i$, leaving a profile likelihood that increases
monotonically in $\kappa$ towards the Poisson limit. A formal statement and
proof are given in \citet{vanoirbeek2026trunc}. Numerically,
fitting~\eqref{eq:loglik} to data generated from the year-level model
with $\kappa = 10$ returns $\hat{\kappa}$ in the hundreds of thousands or larger in
every replication. The dispersion the NB--CL model estimates is
therefore the cell-level $\kappa$ of~\eqref{eq:gamma};
accident-year heterogeneity is identified only through anchored or
hierarchical levels and is outside the scope of the present model.
\end{remark}

\subsection{The $\kappa \to 0$ limit}

The limit $\kappa \to 0$ represents extreme cell-level shock dispersion.
In this regime, the Gamma shocks $\varepsilon_{i,j} \sim \mathrm{Gamma}(\kappa,\kappa)$
of~\eqref{eq:gamma} become increasingly diffuse:
$\mathrm{Var}(\varepsilon_{i,j}) = 1/\kappa \to \infty$, so individual cells
depart arbitrarily far from the multiplicative structure. The cell-level variance
explodes accordingly,
\[
\mathrm{Var}(N_{i,j}) \approx \frac{\mu_{i,j}^2}{\kappa} \to \infty.
\]

The practical implication for reserving is that age-to-age factors become
unreliable. A development factor $\hat{f}_j$ is a ratio of cumulative counts
across accident years. When the counts in any row have near-infinite variance,
the numerator and denominator of this ratio are both highly unstable, and the
resulting factor can take extreme values. Reserve estimates for individual
accident years may diverge widely as a consequence, even if the total reserve
remains finite in expectation.

No borrowing of strength across rows occurs in this model at any value
of $\kappa$: with a free $\alpha_i$ per accident year, each row's level
is estimated from that row alone. Cross-row pooling, i.e.\ the credibility-theoretic content of the Poisson--Gamma hierarchy \citep{buhlmann1970}, requires a hierarchical structure on the
levels and is the subject of the credibility extension
(Section~\ref{sec:extensions}).

The profile likelihood for $\kappa$ also flattens as $\kappa \to 0$, making
estimation unstable in this regime.

\begin{remark}[Multinomial split given the row total]
The conditional-Poisson cell structure implies a multinomial development
split. Conditional on the shock vector
$(\varepsilon_{i,0},\ldots,\varepsilon_{i,J-1})$ and on the row total
$N_i = \sum_j N_{i,j}$, the vector of increments is multinomial with cell
probabilities proportional to $w_j \varepsilon_{i,j}$. In the Poisson
limit $\kappa \to \infty$ the shocks degenerate and the split reduces to
$\mathrm{Multinomial}(N_i, \boldsymbol{w})$: each claim is allocated
independently to development period $j$ with probability $w_j$. This is
the point of contact with multinomial development models such as the
Dirichlet--multinomial framework of \citet{srirams2021}
(Section~7.5), which places the variability on the split probabilities
themselves rather than on cell-level rates.
\end{remark}

\section{Estimation}

\subsection{Log-likelihood}

Under the NB--CL model, the log-likelihood for the observed triangle is
\begin{equation}\label{eq:loglik}
\ell(\boldsymbol{\alpha}, \boldsymbol{\beta}, \kappa) = \sum_{i=1}^{I} \sum_{j=0}^{I-i} \log f_{\text{NegBin}}(N_{i,j}; \mu_{i,j}, \kappa),
\end{equation}
where $f_{\text{NegBin}}$ denotes the negative binomial probability mass function. Using the mean-dispersion parameterisation:
\[
\log f_{\text{NegBin}}(N_{i,j}; \mu_{i,j}, \kappa) = \log \Gamma(N_{i,j} + \kappa) - \log \Gamma(\kappa) - \log(N_{i,j}!) + \kappa \log\!\left(\frac{\kappa}{\kappa + \mu_{i,j}}\right) + N_{i,j} \log\!\left(\frac{\mu_{i,j}}{\kappa + \mu_{i,j}}\right).
\]

\subsection{Identifiability constraints}\label{sec:identifiability}

The log-additive structure $\log \mu_{i,j} = \alpha_i + \beta_j$ is not identifiable without constraints, since adding a constant to all $\alpha_i$ and subtracting the same constant from all $\beta_j$ leaves $\mu_{i,j}$ unchanged. We impose the simplex constraint
\[
  \sum_{j=0}^{J-1} \exp(\beta_j) = 1,
\]
which gives $w_j = \exp(\beta_j)$ the interpretation of a development-pattern probability and makes $\mu_i = \exp(\alpha_i)$ the expected total claim count for accident year $i$.

\paragraph{Reparameterisation from GLM software.}
Standard GLM software (e.g.\ \texttt{MASS::glm.nb}) imposes identifiability via the treatment-contrast convention $\tilde{\beta}_0 = 0$, returning coefficients $(\tilde{\alpha}_i, \tilde{\beta}_j)$ that do not satisfy the simplex constraint. To obtain the simplex-parameterised coefficients used throughout this paper, apply the post-hoc transformation
\[
  S = \sum_{l=0}^{J-1} \exp(\tilde{\beta}_l),
  \qquad
  \alpha_i = \tilde{\alpha}_i + \log S,
  \qquad
  \beta_j = \tilde{\beta}_j - \log S.
\]
This is a one-to-one reparameterisation: the cell means $\mu_{i,j}$, the likelihood, the dispersion estimate $\hat{\kappa}$, and all predictive quantities are unchanged. Only the interpretation of the individual $\alpha_i$ and $\beta_j$ coefficients shifts.
\subsection{Estimation via GLM}

The NB--CL model is a generalised linear model with negative binomial response and log link. Estimation proceeds in two stages. First, for fixed $\kappa$, the conditional MLEs of $(\boldsymbol{\alpha}, \boldsymbol{\beta})$ are obtained by fitting a GLM with negative binomial family and log link. Second, the dispersion parameter $\kappa$ is estimated by maximising the profile likelihood $\ell_p(\kappa) = \ell(\hat{\boldsymbol{\alpha}}(\kappa), \hat{\boldsymbol{\beta}}(\kappa), \kappa)$.

In R, this can be implemented using \texttt{MASS::glm.nb}, which performs both steps automatically.

\subsection{Connection to classical Chain--Ladder estimators}

In the limit $\kappa \to \infty$, the negative binomial distribution reduces to the Poisson, and the NB--CL model becomes a Poisson GLM with log link. In this case, the MLEs of $\alpha_i$ and $\beta_j$ coincide with the classical Chain--Ladder estimators, the fitted values $\hat{\mu}_{i,j}$ reproduce the Chain--Ladder projections, and the variance reduces to $\text{Var}(N_{i,j}) = \mu_{i,j}$.

Thus, the NB--CL model generalises the classical Chain--Ladder method while preserving its point estimates in the Poisson limit.

\subsection{Model selection: testing for overdispersion}\label{sec:modelselection}

Since the Poisson CL model is nested within the NB--CL model as the limiting case $\kappa \to \infty$, formal model selection tools can be applied to assess whether the additional dispersion parameter is warranted.

\paragraph{Likelihood ratio test.}
The null hypothesis $H_0: \kappa = \infty$ (Poisson) can be tested against $H_1: \kappa < \infty$ (NB--CL) via the likelihood ratio statistic
\[
\Lambda = 2\bigl[\ell_{\text{NB}}(\hat{\boldsymbol{\alpha}}, \hat{\boldsymbol{\beta}}, \hat{\kappa}) - \ell_{\text{Pois}}(\hat{\boldsymbol{\alpha}}_0, \hat{\boldsymbol{\beta}}_0)\bigr],
\]
where $\ell_{\text{NB}}$ and $\ell_{\text{Pois}}$ denote the maximised log-likelihoods of the NB--CL and Poisson CL models, respectively. Under $H_0$, the parameter $\kappa$ lies on the boundary of the parameter space, so the standard $\chi^2_1$ reference distribution does not apply. Instead, $\Lambda$ follows a 50:50 mixture of a point mass at zero and $\chi^2_1$ \citep{self1987}, yielding a $p$-value of $\frac{1}{2}\Pr(\chi^2_1 \geq \Lambda)$.

In both empirical illustrations (Section~\ref{sec:empirical}), the test is decisive. For the Australian motor bodily injury count data, $\Lambda = 2{,}550.1$ ($p < 10^{-20}$); for the Taylor--Ashe paid amounts data, $\Lambda = 1{,}902{,}378$ ($p < 10^{-20}$). 

\paragraph{Information criteria.}
AIC and BIC provide complementary model selection tools that do not require boundary corrections:
\[
\text{AIC} = -2\ell + 2p, \qquad \text{BIC} = -2\ell + p \log n,
\]
where $p$ is the number of parameters (the NB--CL model has one additional parameter relative to the Poisson CL). For the Australian count data, $\Delta\text{AIC} = 2{,}548$ in favour of the NB--CL model; for the Taylor--Ashe data, $\Delta\text{AIC} = 1{,}902{,}376$. When $\Delta\text{AIC}$ or $\Delta\text{BIC}$ is close to zero, the simpler Poisson CL model may be preferred on parsimony grounds.

\subsection{Practical considerations}

\subsubsection{Zero and sparse cells}

Real triangles often contain zeros, especially in the tail (late development years, recent accident years). For Poisson and negative binomial GLMs with log link, zero cells are permissible since the negative binomial assigns positive probability to zero. If, however, an entire row or column is zero, separation occurs and the MLE for that $\alpha_i$ or $\beta_j$
diverges to $-\infty$. Even near-zero cells can cause near-separation and inflated standard errors.

\subsubsection{Small triangles}

For small triangles (e.g., $I = 5$ or 6), the number of parameters approaches the number of observations. With $I$ accident-year effects, $I - 1$ development-year effects (one constrained), and one dispersion parameter, we have approximately $2I$ parameters from $I(I+1)/2$ cells. Overfitting is a real risk.

For sparse triangles, practitioners should consider parsimonious models with parametric development curves (reducing the number of $\beta_j$ parameters), hierarchical models that shrink $\alpha_i$ toward a common mean, or ridge regularisation by adding a penalty $\lambda(\sum_i \alpha_i^2 + \sum_j \beta_j^2)$ to the log-likelihood.

\subsubsection{Diagnostics}

We recommend examining the profile likelihood for $\kappa$ to verify that it is unimodal and well-peaked rather than flat or multimodal. The standard errors on $\alpha_i$ should be stable; if they explode for recent accident years, this indicates estimation problems. Pearson residuals should be consistent with negative binomial assumptions, showing no systematic patterns across accident years, development years, or calendar years. Finally, the condition number of the Hessian should be monitored, with values above 1000 suggesting numerical instability.

\subsubsection{Interpreting extreme values of $\hat{\kappa}$}

The dispersion parameter $\kappa$ governs the degree of cell-level overdispersion --- the variability of individual cells around the fitted multiplicative structure --- with smaller values indicating greater overdispersion. Table~\ref{tab:kappa_interp} provides guidance for interpretation.

\begin{table}[!htbp]
\centering
\caption{Interpretation of the dispersion parameter $\kappa$}
\label{tab:kappa_interp}
\begin{tabular}{lll}
\toprule
$\kappa$ & Variance-to-mean ratio & Interpretation \\
\midrule
$\kappa \to \infty$ & $\text{Var}/\mu = 1$ & Poisson (no overdispersion) \\
$\kappa = 50$ & $\text{Var}/\mu = 1 + \mu/50$ & Mild overdispersion \\
$\kappa = 20$ & $\text{Var}/\mu = 1 + \mu/20$ & Moderate overdispersion \\
$\kappa = 5$ & $\text{Var}/\mu = 1 + \mu/5$ & High overdispersion \\
$\kappa = 1$ & $\text{Var}/\mu = 1 + \mu$ & Extreme (geometric distribution) \\
$\kappa \to 0$ & $\text{Var}/\mu \to \infty$ & Variance explodes \\
\bottomrule
\end{tabular}
\end{table}

At $\kappa = 1$, the negative binomial reduces to the geometric distribution, implying that the coefficient of variation of the cell-level shocks is 100\%---individual cells depart strongly from the multiplicative structure.

\begin{remark}[Warning signs for small $\hat{\kappa}$]
In practice, estimates of $\hat{\kappa} < 3$ are rare and should prompt investigation rather than blind application of the model. Such extreme heterogeneity may indicate structural breaks in the portfolio (such as mergers, product changes, or underwriting shifts), unmodelled calendar-year effects that have been absorbed into the dispersion parameter, or data quality issues and coding errors. The simulation study in Section~8 demonstrates that the NB--CL model with bias-corrected $\kappa$ maintains good coverage even at $\kappa = 2$, so the method remains reliable in these extreme cases. Nevertheless, practitioners encountering $\hat{\kappa} < 3$ should examine residual diagnostics, test for calendar-year effects, and consider whether the log-additive structure is appropriate for the portfolio.
\end{remark}

\section{Predictive Distribution and Uncertainty Quantification}\label{sec:pred}

The NB--CL likelihood yields a coherent predictive distribution for future incremental claims and for the total reserve. The plug-in form accounts for process variance only; parameter uncertainty is added via the parametric bootstrap of Section~\ref{sec:bootstrap}.

\subsection{Plug-in predictive distribution}

Conditional on the estimated parameters $(\hat{\boldsymbol{\alpha}}, \hat{\boldsymbol{\beta}}, \hat{\kappa})$, the future incremental counts satisfy
\begin{equation}\label{eq:plugin}
N_{i,j}^{\text{future}} \sim \text{NegBin}(\hat{\mu}_{i,j}, \hat{\kappa}), \quad i + j > I,
\end{equation}
where $\hat{\mu}_{i,j} = \exp(\hat{\alpha}_i + \hat{\beta}_j)$.

The total reserve is
\[
R^{(N)} = \sum_{i+j > I} N_{i,j}^{\text{future}}.
\]
Since future cells are conditionally independent given the parameters, the distribution of $R^{(N)}$ can be computed via convolution or Monte Carlo simulation.

\subsection{Process variance only}

The plug-in predictive distribution~\eqref{eq:plugin} incorporates \emph{process variance}, i.e.\ the inherent randomness in future claim counts given fixed, known parameters. Hence, it quantifies how much future counts would vary around their mean if the true parameter vector $\boldsymbol{\theta} = (\boldsymbol{\alpha}, \boldsymbol{\beta}, \kappa)$ were known exactly.

The parameters, however, are not known. They are estimated from the same
observed triangle $\mathcal{D}$ that is used to make the predictions. A
different realisation of the triangle would yield different estimates
$(\hat{\boldsymbol{\alpha}}, \hat{\boldsymbol{\beta}}, \hat{\kappa})$ and
therefore different predicted reserves. This additional variability, i.e.\ the extent to which the reserve estimate changes as a function of which triangle was observed, is \emph{estimation variance}, and it is ignored by the plug-in
approach.

The decomposition
\begin{equation}\label{eq:vardecomp2}
\mathrm{Var}(R) =
\underbrace{\mathrm{E}[\mathrm{Var}(R \mid \boldsymbol{\theta})]}_{\text{process variance}}
+
\underbrace{\mathrm{Var}(\mathrm{E}[R \mid \boldsymbol{\theta}])}_{\text{estimation variance}}
\end{equation}
formalises this distinction. The plug-in distribution accounts only for the
first term. For typical triangle sizes, estimation variance is of the same
order as process variance and cannot be neglected. Predictive intervals based
solely on process variance will therefore undercover systematically: the true
outstanding counts will fall outside the nominal interval more often than
the stated level implies. The parametric bootstrap developed in
Section~\ref{sec:bootstrap} accounts for both terms simultaneously.

\subsection{Variance decomposition}

The total predictive variance decomposes as shown in
equation~\eqref{eq:vardecomp2} above, mirroring the structure of
\citet{mack1993} but arising here from a fully specified likelihood model
rather than second-moment assumptions alone.

\subsection{Parametric bootstrap}\label{sec:bootstrap}

To incorporate both process and estimation variance, we employ a parametric bootstrap procedure.

\subsubsection{Adjusted profile likelihood and bias correction}\label{sec:biascorrection}

Maximum likelihood estimation of $\kappa$ via the profile likelihood $\ell_p(\kappa)$ is subject to finite-sample bias. The profile likelihood treats the estimated means $\hat{\boldsymbol{\mu}}(\kappa)$ as though they were known, ignoring the information consumed in estimating the $p$-dimensional nuisance parameter $\boldsymbol{\lambda} = (\boldsymbol{\alpha}, \boldsymbol{\beta})$. For run-off triangles where $n = I(I+1)/2$ is modest relative to $p \approx 2I - 1$, this bias can be substantial. Empirically, $\hat{\kappa}_{\text{MLE}}$ tends to overestimate the true $\kappa$, leading to underestimation of variance and hence undercoverage of predictive intervals.

The appropriate correction is the Cox--Reid adjusted profile likelihood \citep{coxreid1987,barndorffnielsen1983}. For a model with scalar interest parameter $\kappa$ and vector nuisance parameter $\boldsymbol{\lambda}$, the adjusted profile likelihood is
\[
\ell_{AP}(\kappa) = \ell_p(\kappa) - \tfrac{1}{2}\log\det\bigl|j_{\boldsymbol{\lambda}\boldsymbol{\lambda}}(\hat{\boldsymbol{\lambda}}_\kappa)\bigr|,
\]
where $j_{\boldsymbol{\lambda}\boldsymbol{\lambda}}(\hat{\boldsymbol{\lambda}}_\kappa)$ is the observed information matrix for $\boldsymbol{\lambda}$ evaluated at the constrained MLE $\hat{\boldsymbol{\lambda}}_\kappa$ for fixed $\kappa$. The correction term $-\frac{1}{2}\log\det|j_{\boldsymbol{\lambda}\boldsymbol{\lambda}}|$ removes the $O(1)$ bias in $\ell_p(\kappa)$ arising from nuisance parameter estimation \citep{coxreid1987}.

For the NB--CL model, the observed information for $\boldsymbol{\lambda}$ at fixed $\kappa$ is
\[
j_{\boldsymbol{\lambda}\boldsymbol{\lambda},kl} = \sum_{i,j} \frac{w(\kappa,\mu_{ij})}{v(\kappa,\mu_{ij})} \cdot \frac{\partial\mu_{ij}}{\partial\lambda_k}\frac{\partial\mu_{ij}}{\partial\lambda_l},
\]
where $v(\kappa,\mu) = \mu + \mu^2/\kappa$ is the NB variance function and $w(\kappa,\mu)$ is the corresponding GLM working weight. The log-determinant of $j_{\boldsymbol{\lambda}\boldsymbol{\lambda}}$ may be written as $\sum_{k=1}^p \log \xi_k(\kappa)$, where $\xi_k(\kappa)$ are its eigenvalues. Differentiating $\ell_{AP}(\kappa)$ with respect to $\kappa$ and setting to zero yields the adjusted score equation
\begin{equation}\label{eq:adjscore}
S_p(\kappa) - \tfrac{1}{2}\frac{\partial}{\partial\kappa}\log\det\bigl|j_{\boldsymbol{\lambda}\boldsymbol{\lambda}}\bigr| = 0,
\end{equation}
where $S_p(\kappa) = \partial\ell_p/\partial\kappa$ is the profile score.

To obtain a tractable closed-form correction, we evaluate the adjustment under the approximation that the design is balanced and that fitted means are not negligible relative to $\kappa$. Under these conditions the NB working weights satisfy $\bar{w}(\kappa) = n^{-1}\sum_{ij}w_{ij}(\kappa) \approx \bar{\mu}\kappa/(\kappa + \bar{\mu})$, and the log-determinant scales as $\log\det|j_{\boldsymbol{\lambda}\boldsymbol{\lambda}}| \approx p\log\bar{w}(\kappa) + C$ where $C$ does not depend on $\kappa$. Differentiating and substituting into~\eqref{eq:adjscore} gives a correction to the profile score of order $p/(2\kappa)$. Solving, the leading-order adjustment to the MLE is
\begin{equation}\label{eq:biascorrection}
\hat{\kappa}_{\mathrm{adj}} = \hat{\kappa}_{\mathrm{MLE}} \cdot \frac{n - p}{n}.
\end{equation}
In the Poisson limit $\kappa \to \infty$, the NB variance function reduces to $v = \mu$, the profile score equation for $\kappa^{-1}$ becomes linear, and the correction factor $(n-p)/n$ is exact, coinciding with the REML correction for Gaussian variance components \citep{patterson1971} and with the degrees-of-freedom correction for the quasi-Poisson dispersion parameter \citep{mccullagh1989}. The correction is therefore not novel in itself; what is novel is its derivation as a Cox--Reid adjustment specific to the NB--CL likelihood.

\begin{remark}[Scope of the approximation]
The derivation above is exact in the Poisson limit and first-order correct for finite $\kappa$. For small $\kappa$ (high overdispersion), the approximation $\bar{w} \approx \bar{\mu}\kappa/(\kappa+\bar{\mu})$ is less accurate and the correction factor $(n-p)/n$ may understate the required adjustment. The simulation study of Section~8 demonstrates empirically that the corrected estimator achieves near-nominal coverage even at $\kappa = 2$, suggesting the approximation is adequate in practice. A fully numerical implementation via direct maximisation of $\ell_{AP}(\kappa)$ is straightforward and removes the approximation entirely; we adopt the closed-form correction for practical convenience.
\end{remark}

For a $10 \times 10$ triangle, $n = 55$ and $p = 19$ (one intercept, nine accident-year effects, nine development-year effects), giving a correction factor of approximately 0.65. Here $p$ counts only the mean (nuisance) parameters $\boldsymbol{\lambda}$; the interest parameter $\kappa$ is excluded from the Cox--Reid adjustment. The parameter count of Section~5.5, which includes $\kappa$, serves the different purpose of overfitting assessment. This correction shrinks $\hat{\kappa}$ toward smaller values, appropriately increasing the estimated variance. The simulation study in Section~8 demonstrates that this correction yields well-calibrated predictive intervals across a wide range of $\kappa$ values.

\subsubsection{Bootstrap algorithm}

\begin{algorithm}
\caption{Parametric Bootstrap for NB--CL with Bias-Corrected $\kappa$}
\label{alg:bootstrap}
\begin{algorithmic}[1]
\STATE Fit NB--CL to observed triangle $\mathcal{D}$; obtain $(\hat{\boldsymbol{\alpha}}, \hat{\boldsymbol{\beta}}, \hat{\kappa}_{\text{MLE}})$
\STATE Compute $\hat{\kappa}_{\text{adj}} = \hat{\kappa}_{\text{MLE}} \cdot (n-p)/n$
\FOR{$b = 1, \ldots, B$}
\STATE \textbf{Step 1: Generate synthetic observed triangle}
\FOR{each observed cell $(i,j)$ with $i + j \leq I$}
\STATE $N_{i,j}^{*(b)} \sim \text{NegBin}(\hat{\mu}_{i,j}, \hat{\kappa}_{\text{adj}})$
\ENDFOR
\STATE \textbf{Step 2: Re-estimate parameters}
\STATE Fit NB--CL to $\{N_{i,j}^{*(b)}\}$; obtain $(\hat{\boldsymbol{\alpha}}^{*(b)}, \hat{\boldsymbol{\beta}}^{*(b)}, \hat{\kappa}_{\text{MLE}}^{*(b)})$
\STATE Compute $\hat{\kappa}_{\text{adj}}^{*(b)} = \hat{\kappa}_{\text{MLE}}^{*(b)} \cdot (n-p)/n$
\STATE \textbf{Step 3: Simulate future cells}
\FOR{each future cell $(i,j)$ with $i + j > I$}
\STATE $N_{i,j}^{*(b)} \sim \text{NegBin}(\hat{\mu}_{i,j}^{*(b)}, \hat{\kappa}_{\text{adj}}^{*(b)})$
\ENDFOR
\STATE \textbf{Step 4: Compute total reserve}
\STATE $R^{(N),*(b)} = \sum_{i+j>I} N_{i,j}^{*(b)}$
\ENDFOR
\RETURN Bootstrap distribution $\{R^{(N),*(b)}\}_{b=1}^B$
\end{algorithmic}
\end{algorithm}

Predictive intervals are obtained as empirical quantiles of the bootstrap distribution. For example, the 95\% predictive interval is $[R^{(N), *(0.025)}, R^{(N), *(0.975)}]$.

\begin{remark}[Bootstrap design and parameter uncertainty]
Algorithm~\ref{alg:bootstrap} propagates estimation uncertainty in
$(\boldsymbol{\alpha}, \boldsymbol{\beta}, \kappa)$ by re-estimating all
parameters on each synthetic triangle, following the standard parametric
bootstrap for predictive intervals \citep{DavisonHinkley1997}. This
approach treats the observed triangle as one realisation of the fitted
model and asks how much reserve estimates would vary across realisations,
capturing both process variance and parameter estimation variance
simultaneously. A theoretically distinct procedure places explicit priors on 
$(\boldsymbol{\alpha}, \boldsymbol{\beta}, \kappa)$ and integrates the 
predictive distribution against the resulting posterior; this Bayesian 
alternative is discussed in Section~\ref{sec:extensions} and is the 
recommended approach for small triangles where the quadratic 
approximation underlying the bias correction is least accurate.
\end{remark}

\subsection{Comparison with the residual (ODP) bootstrap}

The parametric bootstrap described above differs from the residual bootstrap of \citet{england2002} in several ways. It does not rely on the ODP quasi-Poisson variance structure and does not resample Pearson residuals. Instead, it resamples from a fully specified negative binomial distribution and naturally extends to any distributional assumption, such as Tweedie for amounts.

Despite these differences, both approaches share the goal of quantifying
process and estimation variance. The NB--CL bootstrap achieves this with
fewer assumptions and a clearer probabilistic interpretation. A further
practical advantage is that tail quantiles of the reserve distribution, such as the 1-in-100 or 1-in-500 outcome relevant for capital setting, are
estimated by sampling from a fully specified parametric family rather than
from a finite pool of residuals. For a $10 \times 10$ triangle the residual
pool contains only 55 observations, making extreme-quantile estimates
from residual resampling unreliable; the parametric bootstrap does not
share this limitation.

\section{Relationship to Existing Models}

The NB--CL model provides a unifying framework that encompasses several well-known stochastic extensions of the Chain--Ladder method, though the nature of the relationship differs across models.

\subsection{Poisson Chain--Ladder}

The Poisson Chain--Ladder model assumes
\[
N_{i,j} \sim \text{Poisson}(\mu_{i,j}), \quad \log \mu_{i,j} = \alpha_i + \beta_j.
\]
This is a special case of the NB--CL model obtained by letting $\kappa \to \infty$. In this limit, $\text{Var}(N_{i,j}) = \mu_{i,j}$, and the negative binomial reduces to the Poisson.

\subsection{Over-Dispersed Poisson (ODP) Chain--Ladder}

The ODP model assumes
\[
\operatorname{Var}(N_{i,j}) = \phi\, \mu_{i,j},
\]
where $\phi > 1$ is an overdispersion parameter. This variance structure 
is motivated by quasi-likelihood theory rather than a fully specified 
probability distribution.

The relationship between the ODP model and Mack's distribution-free 
Chain--Ladder (DFCL) model has been the subject of extended methodological 
debate. \citet{mackventer2000} provide a detailed comparison and show 
that, although both models reproduce the Chain--Ladder point estimates, 
they differ structurally in their independence assumptions, in the fitted 
values they imply within the observed triangle, and in their behaviour 
when the data structure deviates from a strict triangle (e.g., trapezoidal 
data, missing cells, or alternative weighting of development factors). 
They conclude that only the DFCL model qualifies as the stochastic model 
underlying the Chain--Ladder algorithm. The NB--CL model offers a third 
route that is not subject to their critique: rather than selecting between 
two ad hoc working models on the basis of which are considered closer 
to the deterministic algorithm, it derives the Chain--Ladder structure 
from a coherent micro-level process. The DFCL/ODP debate becomes a 
question of which approximation to the NB--CL likelihood is preferred, 
not which model truly underlies the algorithm.

The technical content of this resolution lies in the variance functions. 
The NB2 parameterisation adopted here implies a quadratic variance-mean 
relationship,
\[
\operatorname{Var}(N_{i,j}) = \mu_{i,j} + \frac{\mu_{i,j}^2}{\kappa},
\]
whereas the ODP model assumes a linear relationship 
$\operatorname{Var}(N_{i,j}) = \phi\,\mu_{i,j}$. These two variance 
structures are structurally distinct: the ODP variance function is not 
a special case of the NB2 variance function for any value of $\kappa$. 
The connection between them is nonetheless precise. Rewriting the NB2 
variance as
\[
\operatorname{Var}(N_{i,j}) = \mu_{i,j}\!\left(1 + \frac{\mu_{i,j}}{\kappa}\right)
\]
shows that, if cell means are approximately equal to some portfolio average $\bar{\mu}$, the NB2 variance coincides with the ODP variance under the identification $\phi = 1 + \bar{\mu}/\kappa$. This is a local approximation valid when the range of $\mu_{i,j}$ across cells is narrow relative to $\kappa$; it breaks down in cells with large counts where the quadratic term dominates.

The ODP model can also be embedded in a proper likelihood framework through the NB1 parameterisation, under which $\operatorname{Var}(N_{i,j}) = \mu_{i,j}(1+\alpha)$ for some $\alpha > 0$, yielding a linear variance-mean relationship that exactly matches ODP. The NB2 formulation adopted in the present paper implies a quadratic relationship instead, which is more appropriate when overdispersion increases with the magnitude of expected counts, a pattern commonly observed in cells with large exposure. The choice between NB1 and NB2 is an empirical question assessable via residual diagnostics or likelihood-based model comparison. 

\subsection{Negative-Binomial model of Verrall (2000)}

A closely related but structurally distinct negative binomial 
model was proposed by \citet{verrall2000} and is presented in 
\citet{wuthrich2008} as a conditional time-series formulation. 
That model specifies:
\begin{equation}
X_{i,j} \mid C_{i,j-1} \;\sim\; 
\mathrm{NB1}\!\left(C_{i,j-1}(f_{j-1}-1),\; f_{j-1}\right),
\label{eq:verrall_nb}
\end{equation}
where $\mathrm{NB1}(m, \delta)$ denotes a negative binomial distribution 
parameterised by its mean $m$ and its variance-to-mean ratio $\delta$, so 
that the variance is \emph{linear} in the mean. This is deliberately 
distinct from the mean-dispersion (NB2) convention 
$\mathrm{NegBin}(\mu, \kappa)$ used throughout this paper, under which the 
variance is quadratic in the mean; the Verrall model cannot be written in 
the NB2 notation. Its first two moments are:
\begin{align}
\mathrm{E}[X_{i,j} \mid C_{i,j-1}] &= C_{i,j-1}(f_{j-1}-1), \\
\mathrm{Var}(X_{i,j} \mid C_{i,j-1}) &= 
C_{i,j-1}(f_{j-1}-1)f_{j-1}.
\end{align}
This differs from the NB--CL model in two structural respects.

First, the dispersion parameter in \eqref{eq:verrall_nb} is 
$f_{j-1}$, which varies by development year. The NB--CL model 
has a single constant dispersion parameter $\kappa$ across all 
cells, reflecting the assumption of a homogeneous cell-level shock distribution: $\kappa$ is a property of the residual claim-generating environment, not of the development pattern. The Verrall (2000) model has no such 
micro-level justification for why the dispersion should equal 
the development factor: the Verrall model is thus an NB1-type specification (cf.\ 
Section~7.2), with a development-year-specific variance-to-mean ratio.

Second, the Verrall (2000) model is formulated conditionally 
on the cumulative $C_{i,j-1}$, making it a time-series model 
in the spirit of the Mack assumptions. The NB--CL model is 
formulated unconditionally at the cell level, with the 
multiplicative mean structure $\log\mu_{i,j} = \alpha_i + 
\beta_j$ arising from the Poisson--Gamma micro-derivation 
rather than from recursive development.

Both models are consistent with the Chain--Ladder assumptions 
\citep{mack1993}. The \citet{verrall2000} model reproduces the 
Chain--Ladder point estimates; the NB--CL estimates coincide with 
them in the Poisson limit and approximate them closely for finite 
$\kappa$ (Section~8.2). The models further differ in their variance 
functions, in the interpretation of the dispersion parameter, and 
in whether the dispersion is constant or development-varying. The NB--CL model is 
preferable when the goal is a structural interpretation of 
overdispersion as cell-level shock variability; the Verrall 
(2000) model may be preferable when the development pattern 
itself is believed to drive the overdispersion.

\subsection{Mack model}

The Mack model \citep{mack1993} is formulated conditionally on cumulative claims:
\[
\E[C_{i,j+1} \mid C_{i,j}] = f_j\, C_{i,j}, \qquad
\mathrm{Var}(C_{i,j+1} \mid C_{i,j}) = \sigma_j^2\, C_{i,j},
\]
with independence across accident years. It specifies only these conditional first and second moments, assumes no full probability distribution, and obtains the mean squared error of prediction analytically rather than by resampling.

The NB--CL model differs in three respects. First, it specifies a full likelihood rather than conditional moments. Second, its second-moment structure is unconditional, constant-dispersion, and quadratic in the cell mean, whereas Mack's is conditional, column-specific, and linear in the cumulative; the two specifications are non-nested, as \citet{mackventer2000} emphasise in the parallel ODP/DFCL comparison. Third, prediction uncertainty in the NB--CL model follows from the predictive distribution itself (Section~6), not from separate variance formulae.

The Mack model is therefore not a moment truncation of the NB--CL model; it is an alternative, non-nested second-moment specification that shares the Chain--Ladder point estimates.

\subsection{Dirichlet--multinomial Chain--Ladder}

The Dirichlet--multinomial Chain--Ladder model places a Dirichlet prior 
on the multinomial probabilities governing the development pattern, 
with a concentration parameter $c$ controlling the variability of the 
development pattern across accident years. \citet{srirams2021} develop 
this framework to unify the Chain--Ladder and Bornhuetter--Ferguson 
methods within a single Bayesian model: CL emerges when conditioning 
on observed row totals, while BF emerges when integrating out the 
ultimates under a strong prior. Marginalising over the Dirichlet prior 
yields a negative binomial distribution for aggregated counts, 
paralleling the Poisson--Gamma route taken in this paper.

The NB--CL model and the Dirichlet--multinomial framework capture 
different sources of heterogeneity and are therefore complementary 
rather than equivalent. The NB--CL dispersion parameter $\kappa$ captures cell-level shock variability around the multiplicative structure (arrival-rate heterogeneity across accident years being absorbed by the free $\alpha_i$; Remark~\ref{rem:absorption}). The Dirichlet 
concentration parameter $c$ captures variability in the development 
pattern itself across accident years, holding the arrival rate fixed. 
Both lead to negative binomial marginals for incremental counts but 
the underlying generative mechanisms are distinct, and a model that 
incorporates both layers simultaneously would extend the present 
framework. The NB--CL model is formulated directly at the incremental 
level and is naturally expressed as a GLM with log link, which 
facilitates standard estimation; the Dirichlet--multinomial formulation 
is more naturally Bayesian.

\subsection{Hierarchical Bayesian Chain--Ladder}\label{sec:taylor}

\citet{taylor2015} surveys Bayesian formulations of the Chain--Ladder
method and develops a unifying framework in which error terms belong to
the exponential dispersion family, with over-dispersed Poisson and
Tweedie errors arising as special cases. Both the Mack and the
cross-classified forms of the Chain--Ladder are treated. In the
cross-classified case, the framework admits priors on row, column or
diagonal parameters; in the Mack case, the Bayes, linear Bayes
(credibility) and MAP estimators are shown to coincide. Estimation
proceeds via MCMC.

The NB--CL model and Taylor's framework are complementary rather than
competing, since both embed the Chain--Ladder within a probabilistic
hierarchy that accommodates accident-year heterogeneity. They differ,
however, in three respects. First, the NB--CL model is derived from a
micro-level Poisson--Gamma construction (Section~\ref{sec:micro}) that
gives the dispersion parameter $\kappa$ a structural interpretation as
the inverse variance of the cell-level shocks; Taylor's framework
specifies the hierarchy directly at the macro level without micro-level
justification. Second, the NB--CL model admits estimation via standard
GLM software (\texttt{MASS::glm.nb} in R), requiring no MCMC
implementation and converging in seconds even for large triangles.
Third, the NB--CL model preserves the classical Chain--Ladder point
estimates exactly in the Poisson limit $\kappa \to \infty$, hereby
providing a clean nesting structure that connects the stochastic and
deterministic versions of the method.

The Bayesian approach of \citet{taylor2015} is preferable when prior
information is available or when the triangle is too sparse for reliable
maximum likelihood estimation. The NB--CL model is better suited as a
default frequentist framework; Section~\ref{sec:extensions} sketches the
natural Bayesian extension.

\subsection{Model hierarchy}

Figure~\ref{fig:hierarchy} illustrates the relationships among Chain--Ladder reserving models. The NB--CL model is directly linked to micro-level Poisson--Gamma processes, induces a negative binomial distribution at the macro level, and admits a GLM formulation. The Poisson CL model is an exact special case ($\kappa \to \infty$). The Mack model is a non-nested second-moment specification that shares the Chain--Ladder point estimates but conditions on cumulatives with column-specific variance parameters (Section~7.4). The ODP CL model is structurally distinct from NB--CL in its variance function but shares the same mean structure, with point estimates coinciding in the Poisson limit; the dashed arrow in Figure~\ref{fig:hierarchy} denotes a local approximation between the two variance functions, valid when cell means are approximately homogeneous (Section~7.2).

\begin{figure}[!htbp]
\centering
\begin{tikzpicture}[
    node distance=1.5cm and 2cm,
    every node/.style={rectangle, draw, text centered, minimum height=1cm, minimum width=2.5cm, font=\small},
    arrow/.style={-{Stealth[length=3mm]}, thick},
    dasharrow/.style={-{Stealth[length=3mm]}, thick, dashed}
]
\node (micro) {Micro-level Poisson--Gamma};
\node (nbcl) [below=of micro] {NB--CL: full likelihood};
\node (pois) [below left=1.5cm and 1cm of nbcl] {Poisson CL};
\node (odp) [below=of nbcl] {ODP CL};
\node (mack) [below right=1.5cm and 1cm of nbcl] {Mack};
\node (det) [below=2cm of odp] {Deterministic CL};

\draw[arrow] (micro) -- node[right, draw=none, minimum height=0, minimum width=0] {\scriptsize aggregation} (nbcl);
\draw[arrow] (nbcl) -- node[left, draw=none, minimum height=0, minimum width=0] {\scriptsize $\kappa \to \infty$} (pois);
\draw[dasharrow] (nbcl) -- node[right, draw=none, minimum height=0, minimum width=0] {\scriptsize local approx.} (odp);
\draw[dasharrow] (nbcl) -- node[right, draw=none, minimum height=0, minimum width=0] {\scriptsize shared point est.} (mack);
\draw[arrow] (pois) -- (det);
\draw[arrow] (odp) -- (det);
\draw[arrow] (mack) -- (det);
\end{tikzpicture}
\caption{Relationships among Chain--Ladder reserving models. Solid arrows denote exact special cases; dashed arrows denote non-nested relationships: a local approximation between the NB2 and ODP variance functions (valid when cell means are approximately homogeneous), and shared point estimates but structurally different, conditional second moments in the Mack case. Neither ODP CL nor the Mack model is a special case of NB--CL. See Sections~7.2 and~7.4 for details.}
\label{fig:hierarchy}
\end{figure}
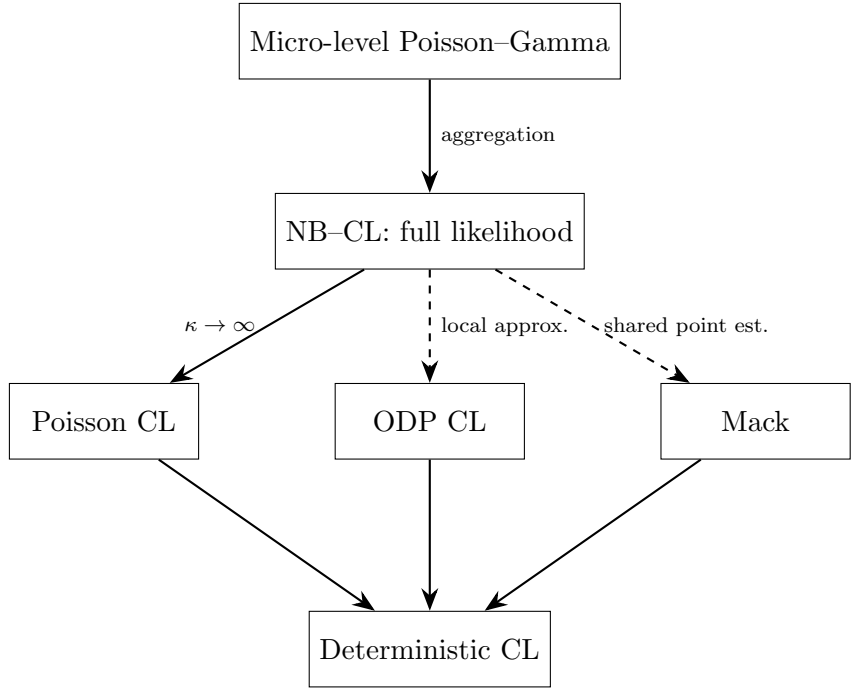

\section{Simulation Study}
\label{sec:simulation}

Three questions are addressed by simulation: the size of the bias correction for $\kappa$ developed in Section~\ref{sec:biascorrection}, the calibration of NB--CL predictive intervals against Poisson CL and
ODP CL under correct specification, and the robustness of the method under three forms of misspecification.

\subsection{Simulation design}

We generate synthetic run-off triangles directly from the NB--CL model to evaluate coverage properties under correct specification.

\subsubsection{Data-generating process}

For each cell $(i,j)$ in the triangle:
\[
  N_{i,j} \sim \operatorname{NegBin}(\mu_{i,j},\, \kappa^{\mathrm{true}}),
  \qquad
  \log \mu_{i,j} = \alpha_i^{\mathrm{true}} + \beta_j^{\mathrm{true}}.
\]
The observed triangle consists of cells with $i + j \leq I$; the remaining cells constitute the future development, and the true outstanding counts are
\[
  R_{\cdot}^{(N),\mathrm{true}} = \sum_{i+j > I} N_{i,j}.
\]

\subsubsection{Parameter settings}

Triangles of size $I = J = 10$ are generated, yielding 55 observed cells and 45 future cells. Accident-year effects are set as $\alpha_i^{\mathrm{true}} = \log\bigl(1000 + 500\,(i-1)/9\bigr)$, representing gradual exposure growth across accident years. The development pattern is $(w_0, \ldots, w_9) \propto (0.53, 0.24, 0.12, 0.05, 0.03, 0.015, 0.005, 0.003, 0.001, 0.001)$, normalised to sum to one; development effects are $\beta_j^{\mathrm{true}} = \log w_j - \log w_0$ so that $\beta_0^{\mathrm{true}} = 0$. For the $I = 7$ triangles of Panel~A the pattern is $(w_0, \ldots, w_6) \propto (0.53, 0.24, 0.12, 0.05, 0.03, 0.02, 0.01)$ with $\alpha_i^{\mathrm{true}} = \log\!\bigl(1000 + 500(i-1)/6\bigr)$, the remaining settings as above. The dispersion parameter takes values
\[
  \kappa^{\mathrm{true}} \in \{2,\, 3,\, 5,\, 10,\, 20,\, 50\},
\]
spanning extreme to mild overdispersion. For each setting, $n_{\mathrm{sim}} = 500$ replications are performed, with $B = 500$ bootstrap samples per replication.

\subsubsection{Methods compared}

We compare four methods that share the same Chain--Ladder predictor structure but differ in their distributional assumptions. \emph{Poisson CL} fits a GLM with Poisson family and therefore ignores any overdispersion present in the data. \emph{ODP CL} fits a GLM with quasi-Poisson family, estimating the overdispersion parameter~$\phi$ from the Pearson residuals; this is the standard approach underlying the ODP bootstrap of \citet{england2002}. \emph{NB--CL (MLE)} fits a GLM with negative binomial family and uses the maximum likelihood estimate~$\hat{\kappa}_{\mathrm{MLE}}$ directly in the bootstrap. Finally, \emph{NB--CL (corrected)} fits the same negative binomial GLM but applies the bias correction~\eqref{eq:biascorrection} to obtain an adjusted~$\hat{\kappa}$ before bootstrapping.

All four methods use the parametric bootstrap described in Algorithm~\ref{alg:bootstrap} to construct predictive intervals, ensuring a fair comparison. The only difference between the two NB--CL variants is whether the degrees-of-freedom adjustment is applied to~$\hat{\kappa}$; comparing them isolates the effect of the bias correction proposed in Section~\ref{sec:biascorrection}.

\subsubsection{Performance metrics}

We assess each method along four dimensions. The first two measure the quality of the point estimate, while the latter two evaluate the predictive intervals.

\emph{Bias} measures the systematic deviation of the point estimate from the true outstanding counts, averaged over all simulations:
\[
  \mathrm{Bias}
  = n_{\mathrm{sim}}^{-1} \sum_{s=1}^{n_{\mathrm{sim}}}
    \bigl(\hat{R}_{\cdot}^{(N),s} - R_{\cdot}^{(N),\mathrm{true},s}\bigr).
\]
\emph{Root mean squared error} (RMSE) captures both bias and variability in a single measure of overall accuracy:
\[
  \mathrm{RMSE}
  = \Biggl(n_{\mathrm{sim}}^{-1} \sum_{s=1}^{n_{\mathrm{sim}}}
    \bigl(\hat{R}_{\cdot}^{(N),s} - R_{\cdot}^{(N),\mathrm{true},s}\bigr)^2\Biggr)^{1/2}.
\]
\emph{Coverage} is the proportion of simulations in which the true outstanding counts fall within the predictive interval. Under correct calibration, the empirical coverage should match the nominal level; systematic undercoverage signals that the predictive distribution is too narrow, while overcoverage indicates unnecessary conservatism. \emph{Interval width} is the average width of the predictive intervals and serves as a complementary diagnostic: among methods that achieve the nominal coverage, narrower intervals indicate sharper inference. Coverage is evaluated at nominal levels of 75\% and 95\%.

\subsection{Results under correct specification}

\begin{table}[!htbp]
\centering
\caption{Simulation results under correct specification 
($n_{\mathrm{sim}} = 500$, $B = 250$ for $I = 7$; 
$n_{\mathrm{sim}} = 500$, $B = 500$ for $I = 10$).
For $I = 7$, only NB--CL (corrected) is reported as the behaviour 
of Poisson CL and ODP CL is qualitatively identical to the 
$I = 10$ results.}
\label{tab:simulation}
\small
\begin{tabular}{clrrrrrr}
\toprule
$\kappa^{\text{true}}$ & Method & Bias & RMSE & Cov.\ 75\% & Cov.\ 95\% & Width 75\% & Width 95\% \\
\multicolumn{8}{l}{\textit{Panel A: $I = 7$, NB--CL (corrected) 
only, $n_{\mathrm{sim}} = 500$, $B = 250$}} \\
\midrule
2  & NB--CL (corrected) &  99 &  871 & 0.73 & 0.93 & 1911 & 3927 \\
3  & NB--CL (corrected) &  39 &  614 & 0.74 & 0.93 & 1394 & 2674 \\
5  & NB--CL (corrected) &   1 &  479 & 0.72 & 0.89 & 1026 & 1869 \\
10 & NB--CL (corrected) &   8 &  324 & 0.70 & 0.93 &  701 & 1219 \\
20 & NB--CL (corrected) &  -2 &  235 & 0.66 & 0.91 &  479 &  821 \\
50 & NB--CL (corrected) &   2 &  154 & 0.63 & 0.88 &  296 &  502 \\
\midrule
\multicolumn{8}{l}{\textit{Panel B: $I = 10$, $n_{\mathrm{sim}} = 
500$, $B = 500$}} \\
\midrule
\multirow{4}{*}{2}
& Poisson CL & 89 & 753 & 0.08 & 0.13 & 134 & 227 \\
& ODP CL & 89 & 753 & 0.57 & 0.82 & 1144 & 2003 \\
& NB--CL (MLE) & 99 & 762 & 0.71 & 0.94 & 1615 & 3156 \\
& NB--CL (corrected) & 99 & 762 & 0.74 & 0.95 & 1666 & 3370 \\
\midrule
\multirow{4}{*}{3}
& Poisson CL & 39 & 559 & 0.08 & 0.15 & 132 & 223 \\
& ODP CL & 39 & 559 & 0.61 & 0.83 & 918 & 1593 \\
& NB--CL (MLE) & 48 & 565 & 0.75 & 0.93 & 1282 & 2385 \\
& NB--CL (corrected) & 48 & 565 & 0.77 & 0.94 & 1322 & 2506 \\
\midrule
\multirow{4}{*}{5}
& Poisson CL & 7 & 415 & 0.12 & 0.20 & 129 & 218 \\
& ODP CL & 7 & 415 & 0.59 & 0.84 & 707 & 1212 \\
& NB--CL (MLE) & 5 & 415 & 0.71 & 0.93 & 933 & 1677 \\
& NB--CL (corrected) & 5 & 415 & 0.74 & 0.94 & 971 & 1755 \\
\midrule
\multirow{4}{*}{10}
& Poisson CL & 13 & 312 & 0.15 & 0.26 & 129 & 220 \\
& ODP CL & 13 & 312 & 0.58 & 0.86 & 509 & 874 \\
& NB--CL (MLE) & 16 & 312 & 0.71 & 0.91 & 659 & 1149 \\
& NB--CL (corrected) & 16 & 312 & 0.74 & 0.91 & 683 & 1194 \\
\midrule
\multirow{4}{*}{20}
& Poisson CL & $-2$ & 214 & 0.24 & 0.38 & 128 & 217 \\
& ODP CL & $-2$ & 214 & 0.60 & 0.88 & 367 & 625 \\
& NB--CL (MLE) & $-1$ & 214 & 0.70 & 0.92 & 447 & 765 \\
& NB--CL (corrected) & $-1$ & 214 & 0.71 & 0.92 & 462 & 795 \\
\midrule
\multirow{4}{*}{50}
& Poisson CL & $-1$ & 143 & 0.33 & 0.55 & 127 & 216 \\
& ODP CL & $-1$ & 143 & 0.63 & 0.85 & 253 & 429 \\
& NB--CL (MLE) & 0 & 142 & 0.66 & 0.88 & 278 & 476 \\
& NB--CL (corrected) & 0 & 142 & 0.67 & 0.89 & 286 & 487 \\
\bottomrule
\end{tabular}
\end{table}

Table~\ref{tab:simulation} summarises the results across all settings. 
Poisson CL severely undercovers, with 95\% coverage rates of 13--55\%,
confirming that ignoring overdispersion leads to grossly inadequate
uncertainty quantification. The England residual bootstrap for ODP CL
achieves 82--88\% at the 95\% level, a substantial improvement over
Poisson but consistently below nominal. NB--CL with the naive MLE
performs better than ODP (88--94\%) but remains just below the corrected
variant at every setting, since the upward bias in
$\hat{\kappa}_{\text{MLE}}$ narrows the intervals.

NB--CL with the bias-corrected $\kappa$ achieves near-nominal 95\%
coverage across all settings, with rates of 89--95\%. Its advantage over
the ODP residual bootstrap is one of calibration rather than sharpness:
at the 95\% level NB--CL covers 4 to 13 percentage
points closer to nominal, at the cost of somewhat wider intervals at low
$\kappa$, the two methods converging in width as overdispersion weakens. The case for NB--CL over ODP therefore does not rest on narrower
or wider intervals, but on the fact that this near-nominal coverage
follows from a coherent full likelihood with a structurally interpretable
dispersion parameter, with the bias correction resolving a genuine
finite-sample estimation problem. At the 75\% level, coverage ranges from
67--77\%, reflecting some shape mismatch between the negative binomial
predictive distribution and the distribution of the total reserve in the
body. Poisson CL and ODP CL share identical point estimates, since both solve
the same quasi-score equations. The NB--CL point estimates differ
slightly for finite $\kappa$, because the NB2 working weights depart from
the Poisson weights (compare the Bias columns); the differences are small
relative to RMSE in all settings. The material differences between the
methods lie in uncertainty quantification.

The bias correction, a single multiplicative factor $(n-p)/n$ grounded in the adjusted profile likelihood framework, is simple to implement and yields the best overall performance: nominal 95\% coverage with intervals appropriately sized for the degree of overdispersion. For mild overdispersion (large $\kappa$), the methods converge, as expected.

Table~\ref{tab:simulation} also reports results for $I = 7$ triangles 
(Panel~A), matching the size of the Australian motor bodily injury 
illustration ($n_{\mathrm{sim}} = 500$, $B = 250$, NB--CL corrected 
only). The bias correction performs adequately but not perfectly at this 
triangle size: 95\% coverage ranges from 0.88 to 0.93 across the panel, 
modestly below nominal throughout, with no clear monotone pattern in 
$\kappa$ (differences of this size are within roughly two Monte Carlo 
standard errors at $n_{\mathrm{sim}} = 500$). The shortfall relative to 
the $I = 10$ results reflects the small-triangle regime: with $n = 28$ 
cells and $p = 13$ mean parameters the adjustment factor 
$(n-p)/n \approx 0.54$ is large, and the quadratic approximation 
underlying the Cox--Reid correction is least accurate. For the 
Australian application, where the corrected estimate is 
$\hat{\kappa}_{\mathrm{adj}} = 2.6$, the nearest panel entries 
($\kappa \in \{2, 3\}$, coverage 0.93) suggest that mild undercoverage 
of one to two percentage points at the 95\% level should be expected. 
The Bayesian formulation discussed in Section~\ref{sec:extensions} 
would resolve this finite-sample gap by propagating parameter 
uncertainty through the posterior rather than approximating it via a 
profile likelihood correction.

\subsection{Robustness under model misspecification}\label{sec:misspec}

Three misspecification scenarios are evaluated below using the NB--CL (corrected) method throughout.

\subsubsection{Scenario A: Poisson DGP}

When the true DGP is Poisson ($\kappa^{\text{true}} = \infty$), the NB--CL model is overparameterised. The key question is whether it gracefully recovers $\hat{\kappa} \to \infty$ and avoids distorting inference.

Table~\ref{tab:misspec} (Scenario~A) shows that the estimated $\hat{\kappa}$ exceeded $10^6$ in all replications, confirming that the profile likelihood correctly identifies the absence of overdispersion. Coverage is 72.5\% at the 75\% level and 96.0\% at the 95\% level---essentially nominal at both levels. The RMSE of 75 (compared to 503 at $\kappa = 10$) reflects the lower inherent variability of the Poisson DGP. The NB--CL model does not degrade when overdispersion is absent: the extra parameter is simply estimated to be large and is effectively inert.

\subsubsection{Scenario B: Unmodelled calendar-year effects}

Calendar-year effects, such as claims inflation, regulatory changes, or operational shifts, are a common source of misspecification for the log-additive model. We introduce a 5\% annual multiplicative inflation factor along diagonals of the triangle, with $\kappa^{\text{true}} = 10$.

Table~\ref{tab:misspec} (Scenario~B) shows that unmodelled calendar-year effects inflate $\hat{\kappa}$ to 16.9 (compared to the true value of 10), as the model partially absorbs the diagonal inflation into the accident-year effects. The reserve bias is negligible (2), because the log-additive structure partially accommodates constant-rate inflation through the accident-year effects. Coverage, however, drops to 71\% at the 75\% level and 92\% at the 95\% level, reflecting the model's inability to fully capture the systematic inflation component. Hence, calendar-year effects should be tested for whenever $\hat{\kappa}$ deviates from expectations; adding a calendar-year term $\gamma_{i+j}$ to the linear predictor (Section~\ref{sec:extensions}) would address this misspecification.

\subsubsection{Scenario C: Development-varying dispersion}

The NB--CL model assumes a single dispersion parameter $\kappa$ for all cells. In practice, overdispersion may be more pronounced in late development years than in early ones. We generate data with $\kappa$ decreasing from 20 in early development years to 3 in late development years: $\kappa_j \in \{20, 18, 15, 12, 10, 7, 5, 4, 3, 3\}$.

Table~\ref{tab:misspec} (Scenario~C) shows that the NB--CL model estimates a single $\hat{\kappa} = 25.5$ on average. The estimate is dominated by the data-rich, high-mean early development columns, where the true $\kappa_j$ is largest and where the likelihood carries most of the information about the dispersion; the finite-sample upward bias of the dispersion MLE (Section~\ref{sec:biascorrection}) pushes the average above even the largest true value $\kappa_j = 20$. Despite this, the 95\% predictive interval attains 95\% coverage. The constant-$\kappa$ assumption is the safest misspecification of the three, because the estimator is dominated by the data-rich early columns while prediction uncertainty is dominated by the data-sparse late columns where the true $\kappa$ is lower. The bias correction effectively compensates for this mismatch.

\begin{table}[!htbp]
\centering
\caption{Simulation results under model misspecification ($I = 10$, $n_{\text{sim}} = 200$, $B = 500$). NB--CL (corrected) is used throughout. The ``Mean $\hat{\kappa}$'' column reports the uncorrected MLE; the bias-corrected value used for the intervals is $\hat{\kappa}_{\mathrm{adj}} = \hat{\kappa}\,(n-p)/n$, which for the reference row averages $\approx 10.8$, close to the true $\kappa = 10$ and explaining the near-nominal coverage.}
\label{tab:misspec}
\begin{tabular}{llrrrrr}
\toprule
Scenario & True DGP & RMSE & Mean $\hat{\kappa}$ & Cov.\ 75\% & Cov.\ 95\% \\
\midrule
A & Poisson ($\kappa = \infty$) & 75 & ${>}10^6$ & 0.73 & 0.96 \\
B & NB ($\kappa = 10$) + calendar (5\%/yr) & 1006 & 16.9 & 0.71 & 0.92 \\
C & NB, $\kappa$ varies by DY (20$\to$3) & 445 & 25.5 & 0.78 & 0.95 \\
\midrule
\multicolumn{6}{l}{\textit{Reference: correct specification ($\kappa = 10$)}} \\
--- & NB ($\kappa = 10$) & 503 & 16.5 & 0.80 & 0.96 \\
\bottomrule
\end{tabular}
\end{table}

\section{Empirical Illustrations}\label{sec:empirical}

The NB--CL model is illustrated on two datasets. The first is a claim count triangle from Australian motor bodily injury insurance, where the micro-level derivation of Section~\ref{sec:micro} applies directly. The second is the Taylor--Ashe paid amounts triangle, included purely as a numerical benchmark; the micro-level derivation does not apply to paid
amounts.

\subsection{Australian motor bodily injury: claim counts}\label{sec:ausbi}

\subsubsection{Data}

The Australian motor bodily injury dataset \citep{charpentier2015} comprises 22{,}036 individual claims from accident years 1993--1999. From these individual records, we construct an incremental claim count triangle with $I = 7$ accident years and $J = 7$ development years, where development year $j$ is defined as the year of claim finalisation minus the accident year. Table~\ref{tab:aus_tri} presents the observed triangle.\footnote{Accident year 1999 contains only 2 claims in development year~0, reflecting the partial final year of the extraction. The cell is retained in the fit; its accident-year effect is estimated from this single observation, which drives the extreme coefficient of variation for 1999 in Table~\ref{tab:aus_reserves}. Refitting without accident year~1999 leaves the dispersion estimate essentially unchanged ($\hat{\kappa}$ moves from 4.8 to 4.6), so the sparse final row does not drive the reported overdispersion.}

\begin{table}[!htbp]
\centering
\caption{Australian motor bodily injury: incremental claim counts}
\label{tab:aus_tri}
\begin{tabular}{lrrrrrrr}
\toprule
AY & DY 0 & DY 1 & DY 2 & DY 3 & DY 4 & DY 5 & DY 6 \\
\midrule
1993 & 220 & 855 & 744 & 414 & 387 & 304 & 44 \\
1994 & 320 & 1{,}133 & 839 & 671 & 480 & 73 & --- \\
1995 & 400 & 1{,}146 & 1{,}141 & 917 & 145 & --- & --- \\
1996 & 347 & 1{,}377 & 1{,}343 & 185 & --- & --- & --- \\
1997 & 357 & 1{,}914 & 299 & --- & --- & --- & --- \\
1998 & 485 & 280 & --- & --- & --- & --- & --- \\
1999 & 2 & --- & --- & --- & --- & --- & --- \\
\bottomrule
\end{tabular}
\end{table}

Cell counts range from 2 to 1{,}914, and the data are genuine integer-valued claim counts. Unlike the Taylor--Ashe paid amounts data, this triangle directly satisfies the micro-level assumptions of Section~4, making the Poisson-mixture micro-level reading of $\kappa$ fully valid at the cell level. The development pattern exhibits a documented structural drift. Measured on a common basis as the first-development-year share within the first two development years, $N_{i,0}/(N_{i,0}+N_{i,1})$ rises from $220/1{,}075 = 20.5\%$ in 1993 to $485/765 = 63.4\%$ in 1998. (The full-row share for 1993, $220/2{,}968 = 7.4\%$, is not comparable to the two-cell 1998 figure.) This drift is structured, across-row non-stationarity rather than i.i.d.\ cell noise. The constant-$\kappa$ model cannot represent it directly, so $\hat\kappa$ absorbs it as residual overdispersion; the empirical $\hat\kappa = 4.8$ therefore conflates genuine cell-level dispersion with this drift, which is exactly why the by-development-year residuals (Figure~\ref{fig:aus_residuals}) show a pattern and motivate the time-varying-$\beta$ extension.

\subsubsection{Model fitting}

The NB--CL model is fitted using \texttt{MASS::glm.nb} with $n = 28$ observed cells and $p = 13$ parameters (intercept, 6 accident-year effects, 6 development-year effects). The estimated dispersion parameter is $\hat{\kappa}_{\text{MLE}} = 4.8$ with a profile likelihood 95\% confidence interval of $[2.8, 7.7]$, indicating substantial overdispersion. The bias-corrected estimate is $\hat{\kappa}_{\text{adj}} = 4.8 \times (28 - 13)/28 = 2.6$. That $\hat{\kappa}_{\text{adj}}$ lies below the profile-likelihood interval is not a contradiction: the interval quantifies sampling uncertainty around the \emph{uncorrected} MLE, whose finite-sample upward bias is precisely what the correction removes.

The likelihood ratio test for overdispersion yields $\Lambda = 2{,}550.1$ with $p < 10^{-20}$, overwhelmingly rejecting the Poisson model. The AIC improvement equals $\Delta\text{AIC} = 2{,}548$, hence the overdispersion parameter is clearly needed.

The dispersion parameter $\hat{\kappa} = 4.8$ has a direct structural interpretation: the variance of the cell-level shocks is $1/\hat{\kappa} = 0.21$, a coefficient of variation of $46\%$ around the fitted multiplicative structure, consistent with the documented instability of the reporting pattern in this portfolio. Under the simplex parameterisation of Section~\ref{sec:identifiability}, the estimated accident-year totals $\hat{\mu}_i = \exp(\hat{\alpha}_i)$ provide a direct readout of the implied ultimate claim count for each accident year, and the development weights $\hat{w}_j = \exp(\hat{\beta}_j)$ sum to unity by construction.

\subsubsection{Reserve estimates}

Table~\ref{tab:aus_reserves} presents accident-year reserve estimates with 95\% prediction intervals from the NB--CL (corrected) bootstrap with $B = 5{,}000$ replications.

\begin{table}[!htbp]
\centering
\caption{Accident-year reserve estimates for Australian motor bodily injury claim counts. NB--CL (corrected) bootstrap with $B = 5{,}000$.}
\label{tab:aus_reserves}
\begin{tabular}{lrrrr}
\toprule
Accident year & CL point estimate & 95\% PI lower & 95\% PI upper & CV (\%) \\
\midrule
1994 & 53 & 2 & 198 & 116.5 \\
1995 & 293 & 44 & 929 & 82.7 \\
1996 & 657 & 128 & 1{,}770 & 69.7 \\
1997 & 1{,}205 & 282 & 3{,}248 & 67.3 \\
1998 & 966 & 342 & 4{,}362 & 67.8 \\
1999 & 17 & 0 & 81 & 121.6 \\
\midrule
Total & 3{,}191 & 1{,}563 & 7{,}785 & 42.5 \\
\bottomrule
\end{tabular}
\end{table}

The coefficient of variation increases for more recent accident years, where a larger proportion of development remains outstanding; it is highest for accident year 1999, which has only one observed cell (2 claims in DY~0). The prediction intervals are highly asymmetric: the upper bound for the total reserve ($7{,}785$) is 2.4 times the point estimate, while the lower bound ($1{,}563$) is 0.5 times, reflecting the right-skewed nature of the negative binomial predictive distribution. This asymmetry would be missed by normal approximations or delta-method intervals.

The total reserve CV of 42.5\% is substantial, driven largely by the high cell-level dispersion ($\hat{\kappa} = 4.8$) around the fitted development structure and the non-stationary development pattern. For comparison, Mack's distribution-free model yields a total reserve standard error of 1{,}612 (CV 50.5\%), reflecting Mack's column-specific variance parameters $\sigma_j^2$; the NB--CL total CV is thus somewhat lower than Mack's on this portfolio. Practitioners should note that the development pattern in this dataset is non-stationary (Table~\ref{tab:aus_tri}), which contributes to the large estimated overdispersion.

\subsubsection{Diagnostics}

Figure~\ref{fig:aus_diagnostics} shows diagnostic plots for the NB--CL fit. The left panel displays Pearson residuals against log-fitted values: residuals are centred around zero with no systematic pattern, supporting the log-additive mean structure. The right panel shows the profile likelihood for $\kappa$. The curve is unimodal and well-peaked around the MLE of~4.8, with the 95\% confidence interval $[2.8, 7.7]$ confirming that $\kappa$ is identifiable despite the small triangle size ($7 \times 7$, 28~cells).

\begin{figure}[!htbp]
\centering
\includegraphics[width=\textwidth]{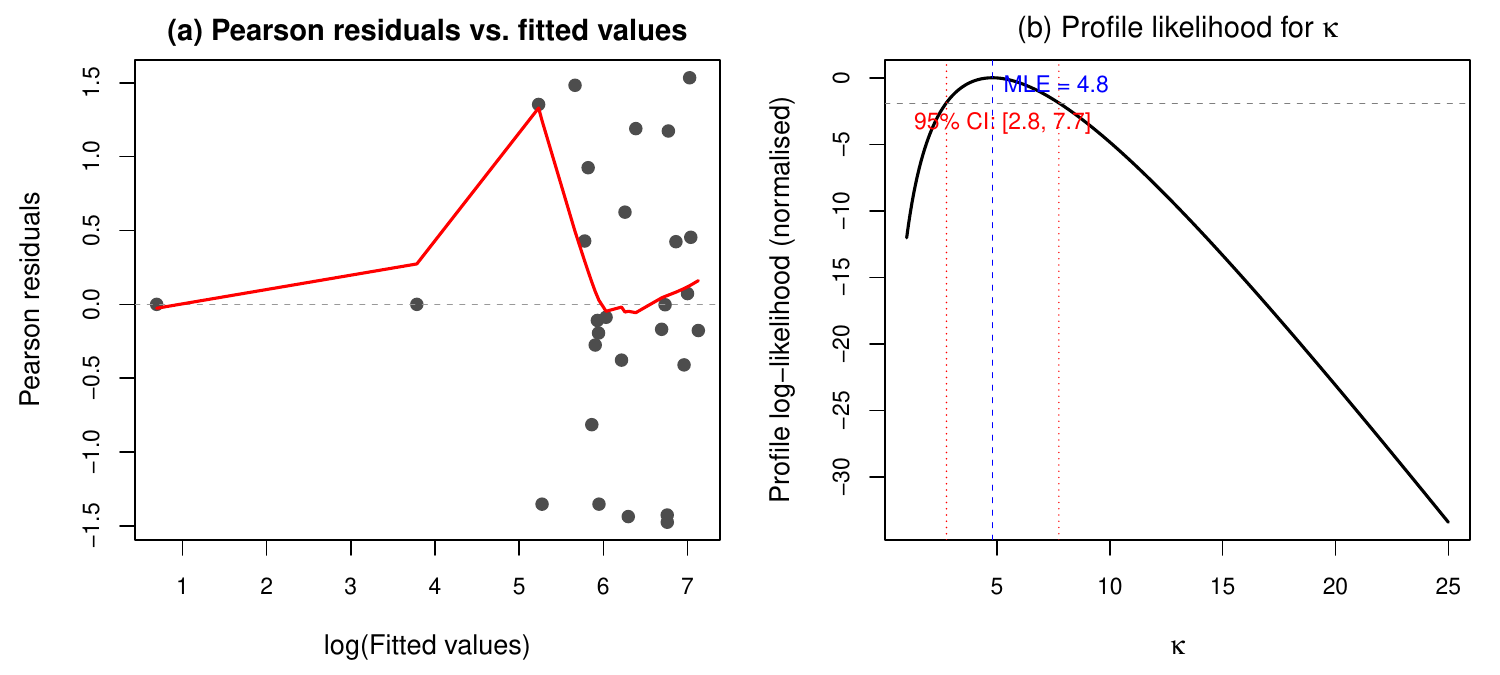}
\caption{Diagnostic plots for NB--CL fit to Australian motor bodily injury count data. (a)~Pearson residuals vs.\ log-fitted values. (b)~Profile likelihood for~$\kappa$ with MLE and 95\% confidence interval.}
\label{fig:aus_diagnostics}
\end{figure}

Figure~\ref{fig:aus_residuals} displays residuals stratified by accident year and development year. The accident-year panel shows no systematic pattern. The development-year panel reveals a mild drift in residual medians across early development periods, which may reflect the non-stationary reporting patterns documented in the data description---the proportion of same-year finalisations increases substantially from 1993 to 1999. This structure would not be captured by the basic log-additive model and motivates the calendar-year extensions discussed in Section~\ref{sec:extensions}.

\begin{figure}[!htbp]
\centering
\includegraphics[width=\textwidth]{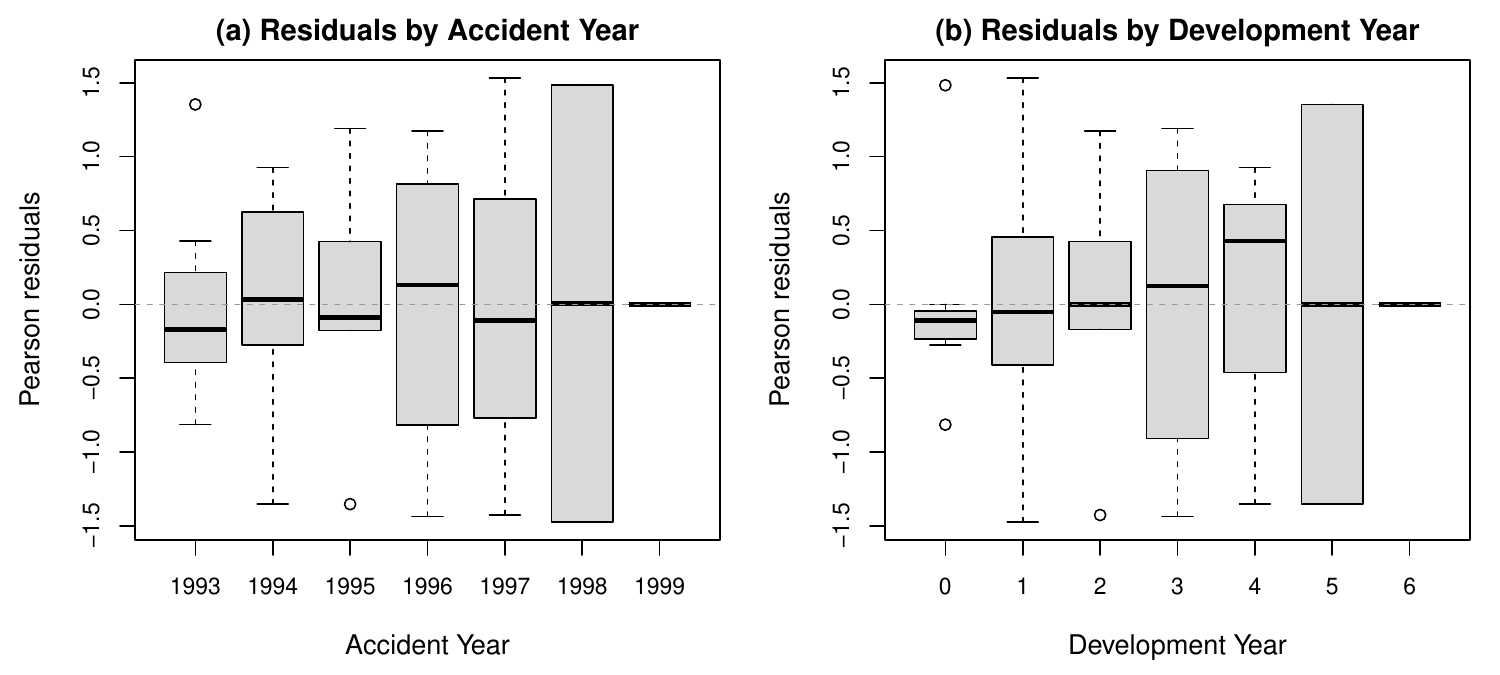}
\caption{Pearson residuals by factor level for Australian motor bodily injury. (a)~By accident year. (b)~By development year.}
\label{fig:aus_residuals}
\end{figure}

\subsection{Benchmark comparison: Taylor--Ashe paid amounts}

The Taylor--Ashe triangle \citep{taylor1983} is the standard benchmark
in the reserving literature and serves here purely for numerical
comparison. The NB--CL model is applied as a working approximation;
the micro-level derivation of Section~4 does not apply to paid amounts,
and a rigorous treatment would require Tweedie GLMs or a
frequency-severity decomposition \citep{wuthrich2008}.

Fitting \texttt{MASS::glm.nb} to the $10 \times 10$ triangle yields
$\hat{\kappa} = 13.8$. The classical Chain--Ladder point estimate, equivalently the Poisson CL and ODP CL estimate, is \$18.68 million
(18{,}680{,}856); the NB--CL point estimate is \$18.09 million
(18{,}085{,}795), a $3.2\%$ departure. This is a concrete instance of the
finite-$\kappa$ point-estimate difference of Section~8.2: the NB2 working
weights downweight the high-mean cells relative to the Poisson weights,
shifting the fitted development structure. The NB--CL (corrected) interval is $[13{,}288{,}238;\, 24{,}447{,}436]$,
compared to the ODP bootstrap interval of $[13{,}400{,}018;\, 25{,}308{,}582]$. The two interval
widths are close (11.2M vs.\ 11.9M), consistent with the simulation finding
(Table~\ref{tab:simulation}) that the gap between NB--CL and the ODP residual
bootstrap narrows as overdispersion weakens; at $\hat\kappa = 13.8$ the two are in
the mild-overdispersion regime. The
structural conclusions of the paper rest entirely on the Australian
motor bodily injury illustration.

\section{Discussion}

\subsection{Interpretability and micro-level consistency}
The NB--CL model sets itself apart from quasi-likelihood and moment-based approaches through its interpretability. The micro-level construction of Section~\ref{sec:micro} derives the negative binomial as the marginal of a Poisson--Gamma process in which the Gamma shock captures cell-level overdispersion in the incremental counts, i.e.\ period-to-period fluctuation in the claim-generating environment around the fitted development structure, rather than heterogeneity across accident years, which the free accident-year parameters absorb. As such, the overdispersion routinely observed in claims triangles is not a statistical nuisance to be absorbed by a dispersion parameter, but a quantitative reading of the portfolio's reporting-pattern instability.  This is illustrated by the Australian motor bodily injury example of Section~\ref{sec:ausbi}: the fitted $\hat\kappa = 4.8$ corresponds to a coefficient of variation of $1/\sqrt{\hat\kappa} \approx 46\%$ in the cell-level shocks around the fitted development pattern. That figure is inflated by the structured reporting-pattern drift documented in Section~\ref{sec:ausbi}, which the constant-$\kappa$ model absorbs rather than represents; what $\kappa$ does \emph{not} measure is heterogeneity across accident-year intensities, which the free accident-year effects absorb.

The ODP and Mack models reproduce the Chain--Ladder point estimates and, in moderate-$\kappa$ regimes, variance estimates of broadly similar magnitude to the NB--CL ones (though on the Australian portfolio Mack's total CV is the larger; Section~\ref{sec:ausbi}). Neither model explains, however, why the assumed variance structure should hold. Note that this is not a
numerical limitation but a structural one: a variance assumption without
a generative model is difficult to validate and to extend.

\subsection{Position within the Chain--Ladder family}

Section~7 positions the Chain--Ladder family relative to the NB--CL model: Poisson CL as the exact $\kappa \to \infty$ limit, ODP as a structurally distinct quasi-likelihood framework agreeing only locally, and Mack as a non-nested conditional second-moment specification sharing the point estimates. The discussion-level content is more modest than a strict hierarchy: within the cross-classified, unconditional branch of the family, the ODP quasi-likelihood and the residual bootstrap are second-moment or simulation-based approximations to the NB--CL predictive distribution; the Mack model belongs to a separate, conditional branch, related through shared point estimates rather than through approximation.

\subsection{Uncertainty quantification}

The NB--CL model provides coherent predictive distributions via the parametric bootstrap. Unlike the residual (ODP) bootstrap, which resamples Pearson residuals, the NB--CL bootstrap resamples from a fully specified distribution. This ensures consistency between the assumed model and the uncertainty quantification procedure.

The variance decomposition~\eqref{eq:vardecomp2} explicitly separates process and estimation variance, mirroring the structure of \citet{mack1993} and \citet{merz2007} but arising from a likelihood-based framework. The accident-year-level prediction intervals in Table~\ref{tab:aus_reserves} illustrate the practical value of this decomposition, showing how uncertainty varies across accident years as a function of the remaining development.

\subsection{Limitations}

The NB--CL model has four limitations.

The micro-level derivation in Section~4, i.e.\ Poisson arrivals with Gamma heterogeneity yielding negative binomial counts, is rigorous for incremental claim counts. When applied to incremental paid amounts, as is common in actuarial practice, the NB--CL model should be viewed as a convenient working approximation rather than a structural likelihood. For a principled treatment of claim amounts, Tweedie GLMs \citep{wuthrich2008} or frequency-severity decompositions provide more appropriate foundations.

The classical Chain--Ladder development factors can be computed in closed form from cumulative column totals, without iterative optimisation. The NB--CL model, formulated as a GLM, requires iterative estimation via IRLS (iteratively reweighted least squares). While modern software makes this computationally trivial, practitioners who value closed-form expressions may find this less appealing. In the Poisson limit $\kappa \to \infty$, however, the NB--CL maximum likelihood estimators coincide with the classical marginal-totals estimators, so the connection is not entirely lost.

The log-additive structure assumes that development-year effects are constant across accident years, which may not hold in portfolios undergoing structural changes, as illustrated by the non-stationary development pattern in the Australian motor bodily injury data (Section~\ref{sec:ausbi}). The model also assumes independence across cells, which may be violated in the presence of calendar-year effects or operational changes. The misspecification study (Section~\ref{sec:misspec}) shows that unmodelled calendar-year effects produce the most concerning degradation in coverage (92\% vs.\ nominal 95\%), while development-varying dispersion is relatively benign (95\%). Finally, the model assumes a single dispersion parameter $\kappa$ for all cells, whereas in practice dispersion may vary by development year or calendar year.

\subsection{Extensions}\label{sec:extensions}

Four natural extensions remain. Allowing
$\beta_j$ to vary across accident years would capture structural changes in
reporting behaviour. Adding a calendar-year term $\gamma_{i+j}$ to the
linear predictor would accommodate superimposed inflation. A Bayesian NB--CL model with priors on $(\boldsymbol{\alpha},
\boldsymbol{\beta}, \kappa)$ would provide posterior inference and natural
regularisation for sparse triangles, and is the theoretically preferred
uncertainty quantification procedure: it propagates parameter uncertainty
through the posterior rather than approximating it via a profile likelihood
correction, hereby eliminating the need for the bootstrap entirely.
Empirical validation, particularly for small triangles where estimation
uncertainty in $(\boldsymbol{\alpha}, \boldsymbol{\beta})$ is large relative
to triangle size, still needs to be carried out. Joint modelling of multiple lines of business via copulas or hierarchical
random effects, and integration with micro- or macro-level severity models,
would together extend the NB--CL frequency model into a fully stochastic
multi-line reserving framework.

The Bayesian NB--CL model is closely related to the exact 
Bayesian model of \citet{wuthrich2008}, Section~4.3.1, who 
place a Gamma prior on the Poisson mean in the ODP framework 
and derive the negative binomial as the exact marginal 
distribution. The NB--CL model can therefore be viewed as 
the frequentist profile likelihood counterpart of their 
Bayesian conjugate analysis: both yield the same marginal 
distribution for incremental counts, but the NB--CL 
dispersion parameter $\kappa$ is estimated by maximising 
the adjusted profile likelihood rather than treated as a 
hyperparameter of the Gamma prior. Under a flat prior on 
$\kappa$, the two approaches coincide asymptotically; for 
small triangles, the Bayesian formulation with an 
informative prior on $\kappa$ is the theoretically preferred 
procedure, as it propagates hyperparameter uncertainty 
through the posterior rather than conditioning on a point 
estimate.

\section{Conclusion}

The classical Chain--Ladder method has no likelihood, and its two main
stochastic companions each fall short on a different front: the Mack model
stops at the second moment, while the ODP framework rests on a
quasi-likelihood variance structure without a generative interpretation.
The NB--CL model fills the latter gap. Incremental counts are modelled as
negative binomial with a log-additive mean, the resulting MLEs coincide
with the Chain--Ladder estimators in the Poisson limit, and the dispersion
parameter $\kappa$ is given a structural reading via the micro-level
Poisson--Gamma construction of Section~\ref{sec:micro}: $\kappa$ is the inverse variance of the cell-level Gamma shocks around
the fitted mean structure, hereby turning what is usually treated as a statistical nuisance into an interpretable portfolio characteristic.

We have seen in Section~7 that, within the Chain--Ladder family, only the Poisson CL model is an exact special case of the NB--CL likelihood: the ODP model approximates it only locally, while the Mack model merely shares its point estimates. This positioning clarifies the assumptions underlying the classical reserving methods.

Estimation of the NB--CL model can be performed using standard GLM techniques, with the dispersion parameter estimated via the adjusted profile likelihood. A parametric bootstrap procedure incorporates both process and estimation variance, yielding well-calibrated predictive intervals. Simulation studies demonstrated that the NB--CL model achieves nominal coverage under correct specification and degrades gracefully under model misspecification, with the exception of unmodelled calendar-year effects, which produce modest undercoverage and should be tested for in practice.

Empirical illustrations on both claim count data (Australian motor bodily injury) and paid amounts data (Taylor--Ashe) confirmed that the model fits both data types without numerical difficulty. The claim count illustration, where the micro-level assumptions hold exactly, demonstrated the structural interpretation of $\kappa$ as cell-level
shock dispersion, with the portfolio's documented reporting-pattern
instability as its visible source. Heterogeneity at the accident-year
level, by contrast, is absorbed by the free accident-year parameters and
requires anchored or hierarchical levels for identification
(Remark~\ref{rem:absorption}).

Time-varying development patterns, calendar-year effects, the Bayesian
formulation, and frequency--severity coupling still need to be investigated;
Section~\ref{sec:extensions} sketches the corresponding entry points.

\section*{Acknowledgements}
The author thanks Michel Denuit for helpful comments on an earlier version of
this paper.

\paragraph{Use of generative AI.} A large language model (Anthropic's Claude)
was used as an assistive tool during the preparation of this manuscript. All
scientific content is the author's own, and the author takes sole
responsibility for it.

\bibliographystyle{apalike}
\bibliography{references}

\appendix
\section{R Code for NB--CL Estimation}\label{app:code}

R functions implementing the methods described in this paper are available in the \texttt{hgr} package at \url{https://github.com/robin-vo/hgr}. The code below provides a self-contained implementation for reproducibility. The listing is a minimal illustration: failed bootstrap refits are counted and reported via the \texttt{n\_failed} attribute rather than silently dropped, so the effective number of bootstrap samples is auditable. The verification scripts accompanying the paper additionally wrap each refit in a hard per-fit timeout and checkpoint results incrementally.

\begin{verbatim}
library(MASS)

triangle_to_long <- function(triangle) {
  I <- nrow(triangle)
  n <- I * (I + 1) / 2
  df <- data.frame(AY = integer(n), DY = integer(n), A = numeric(n))
  k <- 0
  for (i in 1:I) {
    for (j in 1:I) {
      if (i + j <= I + 1 && !is.na(triangle[i, j])) {
        k <- k + 1
        df$AY[k] <- i; df$DY[k] <- j; df$A[k] <- triangle[i, j]
      }
    }
  }
  df <- df[1:k, ]
  df$AY <- factor(df$AY); df$DY <- factor(df$DY)
  df
}

fit_nbcl <- function(triangle) {
  df <- triangle_to_long(triangle)
  fit <- glm.nb(A ~ AY + DY, data = df, link = log)
  reparam_simplex(fit)
}

reparam_simplex <- function(fit) {
  # Convert glm.nb output (treatment contrasts, beta_0 = 0)
  # to the simplex parameterisation: sum_j exp(beta_j) = 1.
  cf          <- coef(fit)
  intercept   <- cf["(Intercept)"]
  alpha_tilde <- c(0, cf[grep("^AY", names(cf))])
  beta_tilde  <- c(0, cf[grep("^DY", names(cf))])
  # Absorb the intercept into alpha_tilde
  alpha_tilde <- alpha_tilde + intercept
  S <- sum(exp(beta_tilde))
  fit$alpha_simplex <- alpha_tilde + log(S)   # mu_i = exp(alpha_i)
  fit$beta_simplex  <- beta_tilde - log(S)    # sum_j exp(beta_j) = 1
  fit$w_j           <- exp(fit$beta_simplex)  # development weights
  fit
}

kappa_corrected <- function(fit) {
  n <- length(fitted(fit))
  p <- length(coef(fit))
  fit$theta * (n - p) / n
}

bootstrap_reserve <- function(fit, triangle, B = 5000, correct = TRUE) {
  I <- nrow(triangle)
  df <- triangle_to_long(triangle)
  n <- nrow(df); p <- length(coef(fit))
  kappa0 <- if (correct) fit$theta * (n-p)/n else fit$theta
  reserves <- rep(NA_real_, B)
  n_failed <- 0L
  for (b in seq_len(B)) {
    df$A_star <- rnbinom(n, size = kappa0, mu = fitted(fit))
    fit_b <- tryCatch(
      glm.nb(A_star ~ AY + DY, data = df, link = log),
      error = function(e) NULL)
    if (is.null(fit_b)) { n_failed <- n_failed + 1L; next }
    kb <- if (correct) fit_b$theta * (n-p)/n else fit_b$theta
    R_b <- 0
    for (i in 2:I) for (j in (I-i+2):I) {
      mu_ij <- predict(fit_b, newdata = data.frame(
        AY = factor(i, 1:I), DY = factor(j, 1:I)), type = "response")
      R_b <- R_b + rnbinom(1, size = kb, mu = mu_ij)
    }
    reserves[b] <- R_b
  }
  out <- reserves[!is.na(reserves)]
  attr(out, "n_failed") <- n_failed  # effective B = length(out)
  out
}
\end{verbatim}

\section{Proof of Poisson--Gamma Mixture Result}\label{app:proof}

\begin{proposition}
Let $N \mid \lambda \sim \text{Poisson}(\lambda)$ and $\lambda \sim \text{Gamma}(\kappa, \kappa/\mu)$. Then $N \sim \text{NegBin}(\mu, \kappa)$.
\end{proposition}

\begin{proof}
The marginal probability mass function of $N$ is
\begin{align*}
\Pr(N = n) &= \int_0^\infty \Pr(N = n \mid \lambda)\, f_\lambda(\lambda)\, d\lambda \\
&= \int_0^\infty \frac{\lambda^n e^{-\lambda}}{n!} \cdot \frac{(\kappa/\mu)^\kappa}{\Gamma(\kappa)} \lambda^{\kappa-1} e^{-\kappa\lambda/\mu}\, d\lambda \\
&= \frac{(\kappa/\mu)^\kappa}{n!\, \Gamma(\kappa)} \int_0^\infty \lambda^{n+\kappa-1} e^{-\lambda(1+\kappa/\mu)}\, d\lambda \\
&= \frac{(\kappa/\mu)^\kappa}{n!\, \Gamma(\kappa)} \cdot \frac{\Gamma(n+\kappa)}{(1+\kappa/\mu)^{n+\kappa}} \\
&= \frac{\Gamma(n+\kappa)}{n!\, \Gamma(\kappa)} \left(\frac{\kappa}{\kappa+\mu}\right)^\kappa \left(\frac{\mu}{\kappa+\mu}\right)^n,
\end{align*}
which is the probability mass function of $\text{NegBin}(\mu, \kappa)$ in the mean-dispersion parameterisation.
\end{proof}

\end{document}